\documentclass[11pt,usenames,dvipsnames]{article}

\usepackage[dvipsnames]{xcolor}
\usepackage{pgfplots}
\usepgfplotslibrary{units}
\usepackage{graphicx}
\usepackage{tikz}
\usetikzlibrary{matrix}
\usetikzlibrary{cd}
\usepackage{amsthm}
\usepackage{mathtools}   
\usepackage{amsfonts}
\usepackage{amsmath}
\usepackage[title]{appendix}

\usepackage{latexsym}

\usepackage{amsfonts,amssymb} 
\usepackage{cancel}

\newtheorem{theorem}{Theorem}[section]
\newtheorem{lemma}[theorem]{Lemma}
\newtheorem{proposition}[theorem]{Proposition}
\newtheorem{remark}[theorem]{Remark}
\theoremstyle{definition}

\theoremstyle{remark}

\usepackage{xcolor}
\usepackage{xparse}
\usepackage{tikz}
\usetikzlibrary{matrix}
\usetikzlibrary{cd}

\definecolor{amber}{rgb}{1.0, 0.49, 0.0}
\definecolor{amethyst}{rgb}{0.45, 0.31, 0.59}
\definecolor{applegreen}{rgb}{0.55, 0.71, 0.0}
\definecolor{asparagus}{rgb}{0.0, 0.42, 0.24}
\definecolor{darkbyzantium}{rgb}{0.36, 0.22, 0.33}
\definecolor{darklavender}{rgb}{0.45, 0.31, 0.59}

\def\hybrid{\topmargin 0pt      \oddsidemargin 0pt
        \headheight 0pt \headsep 0pt
        \voffset=-0.5cm
        \textwidth 6.25in       
        \textheight 9.5in       
        \marginparwidth 0.0in
        \parskip 5pt plus 1pt   \jot = 1.5ex}
\catcode`\@=11
\def\marginnote#1{}

\newcount\hour
\newcount\minute
\newtoks\amorpm
\hour=\time\divide\hour by60
\minute=\time{\multiply\hour by60 \global\advance\minute by-\hour}
\edef\standardtime{{\ifnum\hour<12 \global\amorpm={am}%
        \else\global\amorpm={pm}\advance\hour by-12 \fi
        \ifnum\hour=0 \hour=12 \fi
        \number\hour:\ifnum\minute<10 0\fi\number\minute\the\amorpm}}
\edef\militarytime{\number\hour:\ifnum\minute<10 0\fi\number\minute}

\def\draftlabel#1{{\@bsphack\if@filesw {\let\thepage\relax
   \xdef\@gtempa{\write\@auxout{\string
      \newlabel{#1}{{\@currentlabel}{\thepage}}}}}\@gtempa
   \if@nobreak \ifvmode\nobreak\fi\fi\fi\@esphack}
        \gdef\@eqnlabel{#1}}
\def\@eqnlabel{}
\def\@vacuum{}
\def\draftmarginnote#1{\marginpar{\raggedright\scriptsize\tt#1}}
\def\draftlabel#1{{\@bsphack\if@filesw {\let\thepage\relax
   \xdef\@gtempa{\write\@auxout{\string
      \newlabel{#1}{{\@currentlabel}{\thepage}}}}}\@gtempa
   \if@nobreak \ifvmode\nobreak\fi\fi\fi\@esphack}
        \gdef\@eqnlabel{#1}}
\def\@eqnlabel{}
\def\@vacuum{}
\def\draftmarginnote#1{\marginpar{\raggedright\scriptsize\tt#1}}

\def\draft{\oddsidemargin -.5truein
        \def\@oddfoot{\sl preliminary draft \hfil
        \rm\thepage\hfil\sl\today\quad\militarytime}
        \let\@evenfoot\@oddfoot \overfullrule 3pt
        \let\label=\draftlabel
        \let\marginnote=\draftmarginnote
   \def\@eqnnum{(\theequation)\rlap{\kern\marginparsep\tt\@eqnlabel}%
\global\let\@eqnlabel\@vacuum}  }

\def\numberbysection{\@addtoreset{equation}{section}
        \def\theequation{\thesection.\arabic{equation}}}

\def\underline#1{\relax\ifmmode\@@underline#1\else
        $\@@underline{\hbox{#1}}$\relax\fi}

\def\titlepage{\@restonecolfalse\if@twocolumn\@restonecoltrue\onecolumn
     \else \newpage \fi \thispagestyle{empty}\c@page\z@
        \def\thefootnote{\fnsymbol{footnote}} }

\def\endtitlepage{\if@restonecol\twocolumn \else  \fi
        \def\thefootnote{\arabic{footnote}}
        \setcounter{footnote}{0}}  
\relax

\numberbysection
\hybrid

\newfont{\Bbbb}{msbm7 scaled 1\@ptsize00}

\newcommand{\DDD}{\raise-1pt\hbox{$\mbox{\Bbbb D}$}}

\newcommand{\UUU}{\raise-1pt\hbox{$\mbox{\Bbbb U}$}}

\newcommand{\z}{\raise-1pt\hbox{$\mbox{\Bbbb Z}$}}

\def\beq{\begin{equation}}
\def\eeq{\end{equation}}

\theoremstyle{remark}

\usepackage{amsmath}

\begin{document}

\begin{titlepage}

\setcounter{page}{1}




\title{Dualities for rational multi-particle Painlev\'e systems: Spectral versus Ruijsenaars.}
\author{Ilia Gaiur\thanks{School of Mathematics, University of Birmingham, UK. Email: iygayur@gmail.com} \and and \and Vladimir Rubtsov\thanks{LAREMA, UMR 6093 du CNRS, D\'epartment de Math\'ematics , Universit\'e d'Angers, France,} \,\thanks{ITEP, Theory Division, Moscow, Russia,} \,\thanks{Institie of Geometry and Physics (IGAP), Trieste, Italy. Email: volodya@univ-angers.fr.}}

\maketitle

\maketitle
\vspace{-7cm} \vspace{6.5cm}

\begin{abstract}
The extension of the Painlev\'e-Calogero coorespondence for $n$-particle Inozemtsev systems raises to the multi-particle generalisations of the Painlev\'e equations which may be obtained by the procedure of Hamiltonian reduction applied to the matrix  or non-commutative Painlev\'e systems, which also gives isomonodromic formulation for these non-autonomous Hamiltonian systems.  We provide here dual systems for the rational multi-particle Painlev\'e systems ($P_{\operatorname{I}},P_{\operatorname{II}}$ and $P_{\operatorname{IV}}$) by reduction from another intersection a coadjoint orbit of $GL(n)$ action with the level set of moment map. We describe this duality in terms of  the spectral curve of  non-reduced system in comparison to the Ruijsenaars duality.
\end{abstract}


\end{titlepage}

\section{Introduction}

The story of a {\it duality phenomenon} in the realm of multi-particle systems goes back to Simon Ruijsenaars ideas appeared 30 years ago. He came up with the following question:
to construct action-angle variables for both Calogero-Moser models $A_n-$type models as well as their "relativistic" of difference analogues (known now as Ruijsenaars-Scheider systems).
This story is well-documented (see e.g. \cite{GR}). 
The main tool in his approach is a commutation relation for the corresponding Lax 
matrix and some other explicit matrix function both exhibited in the phase-space variables. 
Diagonalising the initial Lax matrix, he has discovered that the auxiliary matrix can be interpreted as a Lax matrix 
of another multi-particle system whose position coordinates  are given by the eigenvalues of the
initial  Lax matrix, in other words by the {\it action variables} of the first model. This duality
can be described by the following statement:  {\it the action variables of
the first system are the position variables of the second, and vice versa}.  
The simplest example of this  {\it Ruijsenaars duality} is exhibited in the self-duality phenomenon for the rational Calogero model. 
The self-duality of this model admits an easy geometric interpretation
in terms of the Hamiltonian reduction of Kazhdan, Kostant and Sternberg (see \cite{Kazdan}). We shall review and remind their approach below in an appropriate for our aims form.

The basic idea for other Calogero-Moser-Ruijsenaars-Scheider models comes also from a symplectic reduction. Starting  with a "big phase space" (which has basically a Lie algebra or a Lie group) origin we have deal with two commuting families of "free canonical Hamiltonians".
After a suitable reduction one can construct two "natural" models on the reduced phase space.
"Free" Hamiltonians transform in non-trivial many- body Hamiltonians and in variables which correspond to particle positions and their roles are interchanged  in terms of both models. The natural map between the two models of the reduced phase space yields the "duality morphism"

Further Gorsky and Nekrasov have developed the ideas of  \cite{Kazdan} using also an infinite-dimensional setting for the reduction procedure
starting with {\it integrable sectors} of Topological Quantum Fields - $N=2$ Yang-Mills for Calogero–Moser and GWZW for Ruijsenaars-Schneider models (see e.g. \cite{GN}). 

The duality ideas  were formulated {\it per se} in the Integrable Systems realm in \cite{FGNR}. They had specified the Ruijsenaars duality as the {\it Action-Coordinate (AC) duality} (see the above description)
and had proposed numerous examples. They focused their attention on infinite-dimensional phase spaces, on holomorphic setting and (in the first time) they extend the
phase space symmetry, considering the Poisson reduction for Poisson-Lie structures and related {\it Heisenberg double} of Fock-Rosly (see for the excellent presentation \cite{Aud}).

We should remark here that L.Feher with various co-authors put this ideas on rigorous and explicit platform (see \cite{Feh-Klim-1, Feh-Klim-2, Feh-Mar} and \cite{Feh}). His  students  continue to investigate various aspects of the Ruijsenaars duality (see e.g. \cite{Gorbe}).

The authors \cite{FGNR}, apart of the AC-duality definition, had introduced another important duality principle - so called {\it Action-Action (AA) duality}.  This duality is transparent in the case of 
Seiberg-Witten theory and maps the integrable system in itself. Locally this map changes the action variable $I \to I_{dual}$ via the canonical transformation generating function $S$ asscociated with the Lagrangian submanifold $\bold I_{dual} = \frac{\partial S}{\partial \bold I}$ and, vice-versa, $\bold I = \frac{\partial S_{dual}}{\partial \bold I_{dual}}$ for the dual canonical transformation generating function $S_{dual}$.

We make an attempt to extend and to compare both AC and AA dualities in the framework of {\it non-autonomous Hamiltonian systems} taking as an example a class of systems associated with confluented
Painlev\'e transcendents. More precisely,  the other key ingredient of our work is a famous class of Painlev\'e equations and their matrix and non-commutative analogues.

A non-commutative Painlev\'e II equation was firstly appeared in the work of second author and Retakh \cite{RetakhRubtsov}. This abstract non-commutative equation generalizes
{\it matrix Painlev\'e II equation} proposed in \cite{BalandinSokolov} to the case of {\it non-commutative Painlev\'e time} and have (via a non-abelian Toda chain) analogue of "rational" solutions expressed in
Hankel {\it quasi-determinants}, admit isomonodromy-like presentation and were used in a non-commutative version of the Riemann-Hilbert problem \cite{BertolaCafasso1}.

The multi-particle non-autonomous Hamiltonian systems, which generalise Painlev\'e equations from another point of view, were introduced (for Painlev\'e VI) by Manin \cite{Man} and Levin-Olshanetsky \cite{LO00}.

Manin has described a configuration space for Painlev\'e VI as a fibration (a pencil of elliptic curves) $\pi: \mathcal E \to B$  and its solutions as (multivalued) sections of the fibration. He has given an  interpretation of  the correspondent  Painlev\'e VI phase space as a (twisted) sheaf of the relative holomorphic one-forms on $\mathcal E.$ This description gives the identification of the Painlev\'e VI moduli space with an affine space of certain special closed 2-forms on $\mathcal E$. If such a form $\Omega$ from this space corresponds to the Painlev\'e VI equation then the corresponding solutions are the leaves of the lagrangian fibration of $\Omega$.  Manin has proposed  a time-dependent Hamiltonian description of Painlev\'e VI as

$$\frac{dq}{d\tau}= \frac{\partial H}{\partial p},\quad \frac{dp}{d\tau}= - \frac{\partial H}{\partial q},$$
where

$$H=\frac{p^2}{2} - \frac{1}{(2\pi)^2}\sum_{k=0}^3\alpha_k \wp(q +T_k/2,\tau).$$

Here  $T_k$ are points of order two and $\wp$ denotes the Weierstrass $\wp-$function. Levin an Olshantetsky have given a "many-particle"generalization of this Manin description in the framework of their approach to non-autonomous version of Hitchin integrable systems.

For other Painlev\'e transcendents  Takasaki \cite{Takasaki1} has computed (by the confluence procedure) the multi-particle {\it Calogero-Painlev\'e Hamiltonians:}

\begin{equation*}
 \begin{aligned} 
\tilde{H}_{V I} : \quad & \sum_{j=1}^{n}\left(\frac{p_{j}^{2}}{2}+\sum_{\ell=0}^{3} g_{\ell}^{2} \wp\left(q_{j}+\omega_{\ell}\right)\right)+g^{2} \sum_{j \neq k}\left(\wp\left(q_{j}-q_{k}\right)+\wp\left(q_{j}+q_{k}\right)\right) \\ \tilde{H}_{V} :  \quad & \sum_{j=1}^{n}\left(\frac{p_{j}^{2}}{2}-\frac{\alpha}{\sinh ^{2}\left(q_{j} / 2\right)}-\frac{\beta}{\cosh ^{2}\left(q_{j} / 2\right)}+\frac{\gamma t}{2} \cosh \left(q_{j}\right)+\frac{\delta t^{2}}{8} \cosh \left(2 q_{j}\right)\right)+\\ & +g^{2} \sum_{j \neq k}\left(\frac{1}{\sinh ^{2}\left(\left(q_{j}-q_{k}\right) / 2\right)}+\frac{1}{\sinh ^{2}\left(\left(q_{j}+q_{k}\right) / 2\right)}\right) & \\
\tilde{H}_{I V} : \quad & \sum_{j=1}^{n}\left(\frac{p_{j}^{2}}{2}-\frac{1}{2}\left(\frac{q_{j}}{2}\right)^{6}-2 t\left(\frac{q_{j}}{2}\right)^{4}-2\left(t^{2}-\alpha\right)\left(\frac{q_{j}}{2}\right)^{2}+\beta\left(\frac{q_{j}}{2}\right)^{-2}\right)+ \\
&  g^{2} \sum_{j \neq k}\left(\frac{1}{\left(q_{j}-q_{k}\right)^{2}}+\frac{1}{\left(q_{j}+q_{k}\right)^{2}}\right)  
\\
\tilde{H}_{I I I} : \quad & \sum_{j=1}^{n}\left(\frac{p_{j}^{2}}{2}-\frac{\alpha}{4} \mathrm{e}^{q_{j}}+\frac{\beta t}{4} \mathrm{e}^{-q_{j}}-\frac{\gamma}{8} \mathrm{e}^{2 q_{j}}+\frac{\delta t^{2}}{8} \mathrm{e}^{-2 q_{j}}\right)+g^{2} \sum_{j \neq k} \frac{1}{\sinh ^{2}\left(\left(q_{j}-q_{k}\right) / 2\right)} 
\\ 
\tilde{H}_{I I} : \quad & \sum_{j=1}^{n}\left(\frac{p_{j}^{2}}{2}-\frac{1}{2}\left(q_{j}^{2}+\frac{t}{2}\right)^{2}-\alpha q_{j}\right)+g^{2} \sum_{j \neq k} \frac{1}{\left(q_{j}-q_{k}\right)^{2}} \\
\tilde{H}_{I} : \quad & \sum_{j=1}^{n}\left(\frac{p_{j}^{2}}{2}-2 q_{j}^{3}-t q_{j}\right)+g^{2} \sum_{j \neq k} \frac{1}{\left(q_{j}-q_{k}\right)^{2}}
 \end{aligned} 
\end{equation*}

Some of these Hamiltonian systems may be viewed as non-autonomous (deformed) versions of Inozemtsev systems (see \cite{Ino}) , which had appeared as a multi-particle version of so-called  Painlev\'e-Calogero correspondence \cite{LO00}. $\tilde{H}_{\operatorname{VI}}$ assigns to elliptic Inozemtsev system, $\tilde{H}_{\operatorname{V}}$ to trigonometric and $\tilde{H}_{\operatorname{IV}}$ to rational. There are also two more Hamiltonians which have rational interaction potential - multi-particle Painlev\'e I and II. Since that we introduce rational multi-particle Painlev\'e systems as $P_{\operatorname{I}},P_{\operatorname{II}}$ and $P_{\operatorname{IV}}$ multi-particle systems. In present work we deal only with rational multi-particle Painlev\'e models.

Recently, the second author (with M.Bertola and M.Cafasso) has provided a scheme which brings the connection between matrix Painlev\'e equations and Takasaki's Hamiltonians \cite{BertolaCafassoRubtsov}. They showed that these multi-particle generalisation of the Painlev\'e systems may be obtained from the matrix Painlev\'e Hamiltonian systems producing the Hamiltonian reduction procedure \`a la Kazhdan, Konstant and Sternberg \cite{Kazdan} and applying some symplectic map for reduced system. Authors also gave an isomonodromic formulation of multi-particle Painlev\'e  using this reduction which answered Takasaki's question about the existence of zero-curvature representation for these Hamiltonians.



Here we do quick review the procedure of reduction for the matrix Painlev\'e equations, we refer to the original paper for details and also \cite{BaBeTa03} for the basic facts of Hamiltonian group actions. We start from the $2n\times 2n$ of the isomonodromic problem for matrix Painlev\'e equations (see \cite{BertolaCafasso1, Kawakami1, Kawakami2})
\begin{equation}\label{0.1}
  \left\{\begin{array}{l} 
    \frac{\partial}{\partial z} \Phi = \mathcal{A}(\mathbf{q},\mathbf{p}, z, t)\Phi \\ \\
  \frac{\partial}{\partial t} \Phi = \mathcal{B}(\mathbf{q},\mathbf{p}, z, t)\Phi
                                                    \end{array}\right.
\end{equation}
where $\mathbf{q}$ and $\mathbf{p}$ are unknown elements of $\mathfrak{gl}(n)$. The compatibility condition
\begin{displaymath}
 \mathcal{A}_t -  \mathcal{B}_z + [ \mathcal{A}, \mathcal{B}] = 0 
\end{displaymath}
gives non-autonomous Hamiltonian system with Hamiltonian $\operatorname{Tr} H(\mathbf{q}, \mathbf{p}, t)$ and symplectic form $\omega=\operatorname{Tr} \operatorname{d}\mathbf{p}\wedge \operatorname{d}\mathbf{q} = \sum\limits_{i,j} \operatorname{d}p_{ij}\wedge  \operatorname{d} q_{ji}$, where $H(q,p,t)$ is corresponding Hamiltonian from the Okamoto's list
\begin{equation*}
\begin{aligned}
    \mathrm{P}_{\mathrm{VI}}: \quad & H=\frac{q(q-1)(q-t)}{t(t-1)}\left[p^{2}-\left(\frac{\kappa_{0}}{q}+\frac{\kappa_{1}}{q-1}+\frac{\theta-1}{q-t}\right) p+\frac{\kappa}{q(q-1)}\right] \\
    \mathrm{P}_{\mathrm{V}} : \quad & H=\frac{q(q-1)^{2}}{t}\left[p^{2}-\left(\frac{\kappa_{0}}{q}+\frac{\theta_{1}}{q-1}-\frac{\eta_{1} t}{(q-1)^{2}}\right) p+\frac{\kappa}{q(q-1)}\right] \\
    \mathrm{P}_{\mathrm{IV}}: \quad & H=2 q\left[p^{2}-\left(\frac{q}{2}+t+\frac{\kappa_{0}}{q}\right) p+\frac{\theta_{\infty}}{2}\right] \\
    \mathrm{P}_{\mathrm{III}}: \quad & H=\frac{q^{2}}{t}\left[p^{2}-\left(\eta_{\infty}+\frac{\theta_{0}}{q}-\frac{\eta_{0} t}{q^{2}}\right) p+\frac{\eta_{\infty}\left(\theta_{0}+\theta_{\infty}\right)}{2 q}\right] \\
    \mathrm{P}_{\mathrm{II}}: \quad & H=\frac{p^{2}}{2}-\left(q^{2}+\frac{t}{2}\right) p-\left(\alpha+\frac{1}{2}\right) q \\
    \mathrm{P}_{\mathrm{I}} : \quad & H=\frac{p^{2}}{2}-2 q^{3}-t q
\end{aligned}
\end{equation*}
So the phase space of matrix Painlev\'e equations is given by 
\begin{equation}\label{symM}
M = (S,\omega), \quad S = \mathfrak{gl}(n)\times\mathfrak{gl}(n) \simeq \{\mathbf{q} \in \mathfrak{gl}(n),\mathbf{p} \in \mathfrak{gl}(n)\}, \quad \omega=\operatorname{Tr} \operatorname{d}\mathbf{p}\wedge \operatorname{d}\mathbf{q}
\end{equation} 
Symplectic manifold $M$ allows a group action of $Gl(n)$ by conjugation
$$
g \circ (\mathbf{q},\mathbf{p}) = (Ad_g\mathbf{q},Ad_g\mathbf{p}) = (g^{-1} \mathbf{q}g,g^{-1}\mathbf{p}g)
$$
which is Hamiltonian, with moment map given by
\begin{equation}\label{momentMap}
    \mathcal{M}(\mathbf{p},\mathbf{q}) = [\mathbf{p},\mathbf{q}].
\end{equation}
Since matrix Painlev\'e Hamiltonians are traces of the rational Okamoto Hamiltonians, they are invariant under $GL$-action and the moment map $\mathcal{M}$ is conserved quantity. Finally, restricting to the level set of momentum
\begin{equation}\label{0.2}
  [\mathbf{p},\mathbf{q}] = \sqrt{-1} g(\mathbf{1}-v^{T}\otimes v),\quad v=(1,\,1,\, \dots 1)
\end{equation}
diagonilizing coordinate $\mathbf{q}$, and resolving moment map for $\mathbf{p}$ we get multi-particle Hamiltonian systems which coincide up to canonical transformation with the Hamiltonians presented by Takasaki. The isomonodromic problem for reduced system comes from the isomonodromic problem for matrix Painlev\'e equations by the special gauge of the unreduced isomonodromic system.



The aim of this small paper to introduce the dual Hamiltonian systems for $P_{\operatorname{I}},\,P_{\operatorname{II}},\,P_{\operatorname{IV}}$. In this case the duality which raises from considering another point on the orbit of $GL(n)$ action. The existence of isomonodromic problem provides the spectral duality for systems obtained by reduction, which differs from Ruijsenaars (action-angle) duality.

For matrix $P_{\operatorname{II}}$ and $P_{\operatorname{I}}$, we also provide an autonomous version which reductions give integrable autonomizations of Takasaki Hamiltonians and may be viewed as further degeneration of rational Inozemtsev system. 

Let us briefly discuss the structure of paper. In Section 2 we remind the basic objects of our study - two types of dualities for integrable many-body systems and illustrate
both in the most simple and transparent cases. We interpret the spectral duality as a special case of the AA duality and formulate it in a form of the Theorem 2.1 (the results of this Theorem are probably well known as a folklore but we did not find them in the appropriate form.

Section 3 contains our main computational results. Here we obtain dual Hamiltonians for multi-partcle Calogero-Painlev\'e IV, II, I and discuss the self-duality property. It happens that the Painlev\'e IV model obeys the property of self-duality but other two rational multi-particle Calogero-Painlev\'e Hamiltonians do not.

In Section 4 we study an autonomous avatars of rational multi-particle Calogero-Painlev\'e systems and their behaviour under reduction. This is a subject of well-known classical {\it Calogero-Painlev\'e correspondence} of Levin-Olshanetsky-Zabrodin-Zotov. We explicitly write the Lax representation for two non-commutative Painlev\'e models. We close this section with a confluence procedure interpretation of Inozemtsev system degenerations and its symplectomorphic properties.

In Section 5 we observe a relation between two type of reduction for non-commutative integrable models (more precisely, for the matrix modified Korteweg-de Vries (MmKdV) equation, the spectral duality and the Calogero-Painlev\'e correspondence mappings. 

In the last section 6 we indicate few possible new directions to explore the proposed duality in the case of difference many-body systems (Ruijsenaars-Schneider models) and $q-$ Painlev\'e systems.

We have collected in the Appendix some technicalities related to an explicit computation of interactive terms for dual multi-partical Painlev\'e-Calogero Hamiltonians.

\begin{remark}
During the text we will use the following terminology - by the matrix Painlev\'e systems we denote Hamiltonian systems which were obtained by Kawakami in \cite{Kawakami1}. By Calogero-Painlev\'e systems we denote multiparticle systems obtained in \cite{BertolaCafassoRubtsov, Takasaki1}. By the dual systems we mean multi-particle systems dual to Calogero-Painlev\'e which are introduced in the section 3 of this text.
\end{remark}
\vskip 2mm \noindent{\bf Acknowledgements.} The authors are grateful to M. Cafasso, N. Nekrasov, E. Zenkevich and A. Zotov for helpful discussions. Our special thanks to M. Mazzocco who carefully read our draft and made several useful remarks. The research of I.G. is a part of his PhD program studies in the University of Birmingham. He is grateful to his advisor prof. M. Mazzocco for numerous and helpful discussions of the Hamiltonian reduction for autonomous and non-autonomous systems. 
  V.R. was partly supported by the project IPaDEGAN (H2020- MSCA-RISE-2017), Grant Number 778010 and by the Russian Foundation for Basic Research 16-51-53034- 716 GFEN.  I.G. is also grateful to School of Mathematics at University of Birmingham for a financial support. Both authors are grateful to the Russian Foundation for Basic Research for a partial support under the Grant RFBR 18-01-00461.They thank the Euler International Institute (Saint Petersburg) and organizers of the International Conference "Classical and Quantum Integrable Systems - 2019" for invitation and a possibility to report their results during the Conference. I.G. is also grateful to all participants and organizers of seminar "Methods of Classical and Quantum Integrable Systems" at Steklov Mathematical Institute (Moscow) for fruitful discussions.

\section{Ruijsenaars duality and Spectral duality}

In this section we remind some basic facts about Hamiltonian reduction and Ruijsenaars duality for many-body systems. We illustrate  the Ruijsenaars duality in the simplest example of self-dual rational Calogero system. Then we introduce a duality of another  kind  which comes from the reduction of matrix isospectral (or isomonodromic) systems. 
\subsection{Ruijsenaars duality. Rational Calogero-Moser system}
The rational Calogero-Moser may be obtained considering free particle Hamiltonian
$$
H = \operatorname{Tr}\left(\frac{\mathbf{p}^2}{2}\right)
$$
on the symplectic manifold $M$ given by (\ref{symM}). The equations of motion are
$$
\left\{\begin{array}{l}
     \dot{\mathbf{p}}=0 \\
     \dot{\mathbf{q}}=\mathbf{p}
\end{array}\right.
$$
Since $H$ is invariant under conjugation, moment map $\mathcal{M}$ (\ref{momentMap}) is a constant of motion. Fixing the level set of momentum given by
$$
[\mathbf{p},\mathbf{q}] = \sqrt{-1} g(\mathbf{1}-v^{T}\otimes v),\quad v=(1,\,1,\, \dots 1)
$$
we may find matrix $C$, such that $\mathbf{q} = CQC^{-1}$ and $C^{-1}(\mathbf{1}-v^{T}\otimes v)C = \mathbf{1}-v^{T}\otimes v$, where $Q=\delta_{ij}q_i$. Acting by $C$ on $\mathcal{M}$ we get
$$
C^{-1}[\mathbf{p},\mathbf{q}]C = [Q, P] = \sqrt{-1} g(\mathbf{1}-v^{T}\otimes v)
$$
Resolving moment map we obtain the entries of $P$ 
$$
P = p_i\delta_{ij} + (1-\delta_{ij})\frac{\sqrt{-1}g}{q_i-q_j}.
$$
Reduced Hamiltonian and symplectic form take form
$$
H=\operatorname{Tr}\left(\frac{\mathbf{p}^2}{2}\right)=\operatorname{Tr}\left(\frac{P^2}{2}\right)=\sum\limits_i \frac{p_i^2}{2} + \sum\limits_{i<j}\frac{g^2}{(q_i-q_j)^2},\quad \omega = \sum\limits_i\operatorname{d} p_i\wedge \operatorname{d}q_i
$$
and the equations of motion turn to
$$
\left\{\begin{array}{l}
    \dot{P} = [P, F]   \\
    \dot{Q} = P + [Q, F]
\end{array}\right.\quad F=-C^{-1}\dot{C}
$$
We see that $P$ is a Lax matrix for a rational Calogero system, so spectrum of $P$ $(I_1, I_2,\dots, I_n)$ are the first integrals for reduced system. We may also find a matrix $\tilde{C}$ which diagonalizes $\mathbf{p}$ such that $\tilde{C}^{-1}(\mathbf{1}-v^{T}\otimes v)\tilde{C} = \mathbf{1}-v^{T}\otimes v$. Reducing at this point we have dual coordinates
$$
\mathbf{p}=\tilde{C}\Lambda \tilde{C}^{-1},\quad \Lambda = \delta_{ij}I_i,\quad  \mathbf{q}=\tilde{C}\Phi \tilde{C}^{-1},\quad \Phi = \delta_{ij}\phi_i + (1-\delta_{ij})\frac{\sqrt{-1}g}{I_i-I_j}
$$
and Hamiltonian reduces to the
$$
H=\operatorname{Tr}\left(\frac{\mathbf{p}^2}{2}\right) = \sum\limits_j I_j,\quad \omega = \sum\limits_j \operatorname{d}I_j\wedge\operatorname{d}\phi_j
$$
which gives action-angle variables $(I_i,\phi_i)$. This coincides with the fact that eigenvalues of $P$ are constants of motion (actions). 

To describe Ruijsenaars duality we have to consider the flow generated by the dual free Hamiltonian
$$H_{\operatorname{(dual)}} = \operatorname{Tr}\left(\frac{\mathbf{q}^2}{2}\right)$$

It is obvious that the first set of reduced variables $\{q_1,q_2,...,q_n;\, p_1,p_2,...,p_n\}$ are the action angle variables. On the other hand for the dual coordinates we have
$$H_{\operatorname{(dual)}}= \sum\limits_i \frac{\phi_i^2}{2} + \sum\limits_{i<j}\frac{g^2}{(I_i-I_j)^2}$$

So the action angle variables for the Hamiltonian $H=\operatorname{Tr}(\mathbf{p}^2/2)$ are the coordinates and actions for the Hamiltonian $H_{\operatorname{(dual)}} = \operatorname{Tr}(\mathbf{q}^2/2)$ and vice versa. This duality is called Ruijsenaars duality or action-coordinate duality. Since the dual systems is a rational Calogero in variables $(I_i, \phi_i)$ the rational Calogero is a self-dual system.

In further text we will use $C$ for the matrix which diagonalizes $\mathbf{q}$ and $\tilde{C}$ for the matrix which diagonalizes $\mathbf{p}$, such that $C^{-1}(\mathbf{1}-v^{T}\otimes v)C = \tilde{C}^{-1}(\mathbf{1}-v^{T}\otimes v)\tilde{C}= \mathbf{1}-v^{T}\otimes v$ . We also will use the following notation for reduced coordinates
$$
Q=C^{-1}\mathbf{q}C=\delta_{ij}q_i,\quad P = C^{-1}\mathbf{p}C=\delta_{ij}p_i + (1-\delta_{ij})\frac{\sqrt{-1}g}{q_i-q_j}
$$
and for dual coordinates
$$
\Lambda = \tilde{C}^{-1}\mathbf{p}\tilde{C} = \delta_{ij}I_i,\quad  \Phi = \tilde{C}^{-1}\mathbf{q}\tilde{C} = \delta_{ij}\phi_i + (1-\delta_{ij})\frac{\sqrt{-1}g}{I_i-I_j}
$$
in the next sections. The fact that we may find such $C$ and $\tilde{C}$ means that the $GL(n)$-action orbit of each point $(\mathbf{q},\mathbf{p})$ from the level set of moment map $\sqrt{-1} g(\mathbf{1}-v^{T}\otimes v)$ intersects this set at least in two more points where $\mathbf{q}$ or $\mathbf{p}$ are diagonal, see figure \ref{fig1} for a sketch of intersection of group action orbit and momentum level set.


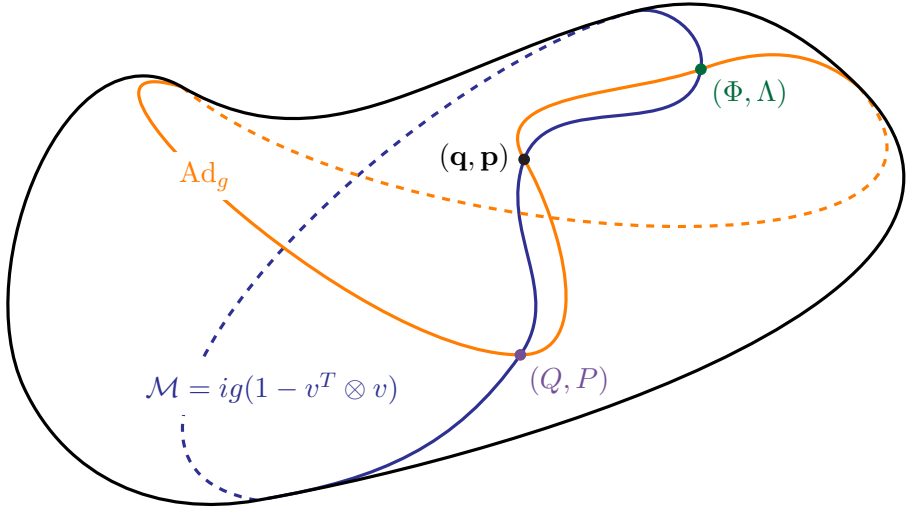
\begin{figure}
\centering
    \begin{tikzpicture}
 \node[below]  (x0) at (0.9,0) {};
 \node[below]  (x1) at (3,4) {};
 \node[below]  (x2) at (6,3.9) {};
 \node[below]  (x3) at (9,5) {};
 \node[below]  (x4) at (12,4) {};
 \node[below]  (x5) at (4,-1.5) {};
 
 \node[very thick] (red) at (7.5,0.3) {};
 \node[color=amethyst, above] (redCoord) at (8.15,-0.37) {$(Q, P)$};
 
 \node[color=asparagus, very thick] (dual) at (9.9,4.1) {};
 \node[color=asparagus, above] (dualCoord) at (10.55,3.45) {$(\Phi, \Lambda)$};
 
 \node[very thick] (start) at (7.55,2.9) {};
 \node[above] (startCoord) at (6.89,2.6) {$(\mathbf{q}, \mathbf{p})$};

 \draw[color=Blue, very thick] 
    (dual.center) to [out=80, in=15](x3.center)
    (dual.center) to [out=-100, in=70](start.center)
    (start.center) to [out=-110, in=55](red.center)
    (x5) to [out=10, in=-125](red.center);
 \draw[color=Blue, dashed, very thick] (x5.center) to [out=170, in=-165](x3.center);
 \draw[color=amber, dashed, below, very thick] (x1.center) to [out=-40, in=310](x4.center);
 \draw[color=amber, below,  very thick] 
    (red.center) to [out=180, in=-197](x1.center)
    (dual.center) to [out=20, in=140](x4.center)
    (dual.center) to [out=200, in=120](start.center)
    (start.center) to [out=-60, in=0](red.center);
 
 \draw[very thick] 
  (x5.center) to [out=10,in=-45](x4.center)
  (x0.center) to [out=-65,in=190](x5.center)
  (x0.center) to [out=115,in=150](x1.center)
  (x1.center) to [out=-30,in=200](x2.center)
  (x2.center) to [out=20,in=195](x3.center)
  (x3.center) to [out=15,in=135](x4.center);
 
 \filldraw [asparagus, above] (9.9,4.1) circle (2pt);
 \filldraw [Black, above] (7.55,2.9) circle (2pt); 
 \filldraw [amethyst, above] (7.5,0.3) circle (2pt);
 \node[color=Blue, below, fill=white!20] (m) at (4.2,0.2) {$\mathcal{M}=ig(1-v^{T}\otimes v)$};
 \node[color=amber,below,fill=white!20] (ad) at (3.3,3) {$\operatorname{Ad}_{g}$};
\end{tikzpicture}
\caption{Intersection of orbit of point $(\mathbf{q},\mathbf{p})$ with level set $\mathcal{M}$. 
\textcolor{Blue}{\rule[1.8pt]{17pt}{0.7mm}} - level set of momentum, \textcolor{amber}{\rule[1.8pt]{17pt}{0.7mm}} - orbit of group action.}
\label{fig1}
\end{figure}

\subsection{Spectral duality}

In a contrast to Ruijsenaars duality, we want to introduce another type of duality which we call spectral duality. This kind of duality is a special case of action-action duality and we formulate it in the next theorem
\begin{theorem}
Let $H(\mathbf{q},\mathbf{p})$ is a "rational" function (i.e. polynomial of the $\mathbf{q},\mathbf{p}$ and their inverses) on symplectic manifold  $M = \{(\mathbf{q},\mathbf{p})\in$  $\mathfrak{gl}(n)\times\mathfrak{gl}(n)\}$ and $\omega = \operatorname{Tr}\operatorname{d} \mathbf{p}\wedge \operatorname{d}\mathbf{q}$, which is invariant under adjoint $GL$-action. If Hamilton equations for  $H(\mathbf{q},\mathbf{p})$ have isospectral representation 
$$
L_t=[L,M],\quad L(\lambda),M(\lambda) \in \mathfrak{gl}(n)\otimes  \mathfrak{gl}(m)
$$
or isomonodromic representation
$$
L_t - M_\lambda + [L,M] = 0,\quad L(\lambda),M(\lambda) \in \mathfrak{gl}(n)\otimes  \mathfrak{gl}(m)
$$
where entries of $L$ and $M$ are "rational" functions of $\mathbf{q}$ and $\mathbf{p}$, then dual the systems $H(p_i,q_i)$ and $H(I_i,\phi_i)$ obtained by reduction from the level set of momentum
$$
[\mathbf{p},\mathbf{q}]=\sqrt{-1}g\left(\mathbf{1}-v^{T}\otimes v\right)
$$
have isospectral (isomonodromic) representation given by the gauge transform for $L$,$M$ operators. Furthermore, spectral curves for both systems are the same
$$
\Gamma^{(p,q)}(\lambda, \mu) = \Gamma^{(I,\phi)}(\lambda,\mu) = \operatorname{det}(L(\lambda) - \mu) = 0.
$$

\end{theorem}
\begin{proof}

First of all we show that $[\mathbf{p},\mathbf{q}]$ is indeed moment map and we may set it to be $ig\left(\mathbf{1}-v^{T}\otimes v\right)$. Vector field generated by the adjoint action of $GL$ on $M$ takes form
$$
    X_\xi = \frac{d}{dt}\left(e^{t\xi}\mathbf{q}e^{-t\xi}, e^{t\xi}\mathbf{p}e^{-t\xi}\right)\Big\vert_{t=0} = ([\xi,\mathbf{q}], [\xi,\mathbf{p}]) = \sum_{ij} [\xi, \mathbf{q}]_{ij}\frac{\partial}{\partial  \mathbf{q}_{ij}} + [\xi, \mathbf{q}]_{ij}\frac{\partial}{\partial  \mathbf{p}_{ij}}
$$
Inserting $X_\xi$ into symplectic form $\omega =  \operatorname{Tr}\operatorname{d} \mathbf{p}\wedge \operatorname{d}\mathbf{q}$ we get
$$
\omega(X_\xi, \circ) = \sum_{ij} -[\xi, \mathbf{q}]_{ij} \operatorname{d}\mathbf{p}_{ji} + [\xi, \mathbf{q}]_{ij}\operatorname{d}\mathbf{q}_{ji} = \operatorname{Tr}\left([\xi, \mathbf{q}]\operatorname{d}\mathbf{p}-[\xi, \mathbf{p}]\operatorname{d}\mathbf{q}\right)=\operatorname{d} \operatorname{Tr}\left(\xi [\mathbf{q},\mathbf{p}]\right)=-\operatorname{d} H_\xi
$$
So by definition (see \cite{Aud}) moment map has form
$$
H_\xi = \langle\mu, \xi \rangle,\quad \mu = [\mathbf{p},\mathbf{q}]
$$
Since Hamiltonian $H(\mathbf{p},\mathbf{q})$ is invariant under adjoint action, the moment map $\mu$ is a constant of motion with respect to the dynamics generated by $H(\mathbf{p},\mathbf{q})$. Since that, we may fix moment map to be
$$
[\mathbf{p},\mathbf{q}]=ig\left(\mathbf{1}-v^{T}\otimes v\right)
$$
and do the reduction diagonilising $\mathbf{q}$ and $\mathbf{p}$ using the same arguments as before. To obtain Lax (or Isomonodromic representation) we do the gauge transform for the eigenfunction of isospectral (isomnodromic) problem
$$
\Psi = (C\otimes \operatorname{Id}_m) Y = gY
$$
where $C$ is the diagonalising matrix of $\mathbf{q}$ or $\mathbf{p}$. That leads to the following action on the $L$ and $M$ matrices
$$
L\rightarrow g^{-1}Lg,\quad M\rightarrow g^{-1}Mg-g^{-1}g_t.
$$
Since gauge transform acts trivially on the second space of the tensor product, we get that each block element of the $L$ matrix conjugated by $C$. Since entries of $L$ and $M$ are also rational we obtain that entries of $L$ depends only on reduced coordinates, which finally gives the isospectral (isomonodromic) operator for reduced system. To obtain $M$ matrix we need to write explicitly $g^{-1}g_t$ which may be done with the help of reduced Hamilton equations (for a detailed proof we refer to the \textbf{Lemma 2.5} in \cite{BertolaCafassoRubtsov}).

Finally we have that $L$ matrix for the reduced system is given by conjugation of $L$ matrix for unreduced system, so they have the same eigenvalues and finally the same spectral curves. Since both reduced ($p,q$) and dual ($I,\phi$) systems are obtained by this procedure they spectral curves coincide with unreduced one and with each other.
\end{proof}
\begin{remark}
The difference between isospectral and isomonodromic cases, that for isospectral case the same spectral invariants are integrals of motion for reduced systems. For isomonodromic duality invariants are not conserved quantities since the dynamic is non-autonomous.
\end{remark}
Here we provide the simple example of kind of duality using matrix harmonic oscillator, since free particle Hamiltonian gives trivial results. In \cite{NN97} duality for Calogero-Moser-Sutherland systems is described via Hamiltonian reduction of the following Hamiltonian system
\begin{equation}\label{simSp}
H=\operatorname{Tr}\left(\frac{\mathbf{p}^2}{2}+\omega^2 \frac{\mathbf{q}^2}{2}\right)
\end{equation}
which we call matrix harmonic oscillator. The construction of dual system more complicated than in case of rational Calogero-Moser, because it requires polar decomposition to change the phase space from $T^{\star}\mathfrak{g}$ to the $T^\star G$, where $G$ is a Lie group and $\mathfrak{g}$ is corresponding Lie algebra. Obtained duality gives Ruijsenaars duality between Calogero and rational Ruijsenaars-Sneider systems.

On the other hand the Hamiltonian system given by (\ref{simSp}) may be written as isospectral deformation of block-matrix $L$-operator
$$
\left\{\begin{array}{l}
    L\Psi = \lambda\Psi \\
    \Psi_t = M\Psi
\end{array}\right.\quad L = \left[\begin{array}{cc}
    \mathbf{p} & \omega \mathbf{q}  \\
    \omega\mathbf{q}  & -\mathbf{p} 
\end{array}\right], \quad M = \frac{\omega}{2}\left[\begin{array}{cc}
    0 & -\operatorname{Id}_n \\
    \operatorname{Id}_n  & 0
\end{array}\right]
$$
The Lax equation $\dot{L} = [L, M]$ is equivalent to the Hamilton-Jacobi equations for the Hamiltonian (\ref{simSp}) with symplectic form $\omega = \operatorname{Tr}\operatorname{d}\mathbf{p}\wedge\operatorname{d}\mathbf{q}$. Hamiltonian $H$ is invariant under conjugation, so we may restrict to the level set of moment map
$$
[\mathbf{p}, \mathbf{q}] = \sqrt{-1}g(\mathbf{1}-v^{T}\otimes v)
$$
Reducing at $Q, P$ we get the system
$$
H^{(\operatorname{red})} = \sum\limits_i\left(\frac{p_i^2}{2} + \omega^2 \frac{q_i^2}{2}\right) + \sum\limits_{i<j}\frac{g^2}{(q_i-q_j)^2}
$$
The Lax pair for a reduced system is given by gauge
\begin{equation}\label{gauge}
\Psi = (C\otimes \operatorname{Id}_2) Y = gY
\end{equation}
which gives the following Lax pair
$$
\left\{\begin{array}{l}
    L^{(\operatorname{red})}Y = \lambda Y\\
    Y_t = (g^{-1}Mg-g^{-1}g_t)Y
\end{array}\right.\quad L^{(\operatorname{red})} = g^{-1}Lg =  \left[\begin{array}{cc}
    P & \omega  Q  \\
    \omega Q  & -P
\end{array}\right]
$$
Reducing at the $\Lambda, \Phi$ we obtain dual system
$$
H^{(\operatorname{dual})} = \sum\limits_i\left(\frac{I_i^2}{2} + \omega^2 \frac{\phi_i^2}{2}\right) + \sum\limits_{i<j}\frac{\omega^2 g^2}{(I_i-I_j)^2}
$$
with Lax pair
$$
\Psi = (\tilde{C}\otimes \operatorname{Id}_2)\tilde{Y} = \tilde{g}\tilde{Y} \quad
\left\{\begin{array}{l}
    L^{(\operatorname{dual})}\tilde{Y} = \lambda \tilde{Y}\\
    \tilde{Y}_t = (\tilde{g}^{-1}M\tilde{g}-\tilde{g}^{-1}\tilde{g}_t)\tilde{Y}
\end{array}\right.\quad L^{(\operatorname{dual})} = \tilde{g}^{-1}L\tilde{g} =  \left[\begin{array}{cc}
    \Lambda & \omega  \Phi  \\
    \omega \Phi  & -\Lambda
\end{array}\right]
$$
Since Lax operators $L,\,L_{(\operatorname{dual})},\,L_{(\operatorname{red})}$ are conjugate to each other the spectral curves for unreduced, reduced and dual systems are the same. Indeed
\begin{multline*}
    \Gamma (\lambda, \mu) = \operatorname{det}(L-\mu) = \operatorname{det}(g)\operatorname{det}(L^{(\operatorname{red})}-\mu)\operatorname{det}(g)^{-1} = \operatorname{det}(L^{(\operatorname{red})}-\mu) =\\= \Gamma^{(\operatorname{red})}(\lambda, \mu) = \operatorname{det}(\tilde{g})\operatorname{det}(L^{(\operatorname{dual})}-\mu)\operatorname{det}(\tilde{g})^{-1} = \Gamma^{(\operatorname{dual})}(\lambda, \mu) = 0
\end{multline*}
In this case there is no dependence on $\lambda$ in spectral curve $\Gamma(\lambda,\mu)=\Gamma(\mu)$, because we consider Lax pair without spectral parameter, in general case the spectral curve depends on both $\lambda$ and $\mu$. Since the spectral curves are the same
$$
 \Gamma^{(\operatorname{red})}(\lambda, \mu) = \Gamma^{(\operatorname{dual})}(\lambda, \mu) =  \Gamma (\lambda, \mu) = 0
$$
we call this spectral duality. This duality may be viewed as follows - if we solve Cauchy problem for unreduced system, the initial data will fix the coefficients of the spectral curve and fix the same data for reduced and dual systems. 

In case of free particle Hamiltonian $\operatorname{Tr}(\mathbf{p}^2/2)$ this spectral duality is trivial - we obtain rational Calogero system and free particle system which corresponds to the action angle variables. In this case this spectral duality is obvious. Consideration of matrix harmonic oscillator raise to nontrivial example of two systems which have the same spectral invariants. 

Obtained reduced system is self-dual in a sense that symplectic map
$$
I_i\rightarrow \omega q_i,\quad \phi_i\rightarrow -\frac{p_i}{\omega}
$$
maps $H^{\operatorname{(red)}}$ to $H^{\operatorname{(dual)}}$, in other words
$$
H^{\operatorname{(dual)}}(q_i,p_i) = H^{\operatorname{(red)}}\left(-\frac{p_i}{\omega}, \omega q_i\right)
$$
This symplectic self-duality comes from symmetry of unreduced Hamiltonian. In general case, the dual systems may not be self-dual.

The close type of duality was introduced in works \cite{MMRZZ131, MMRZZ132}, which provides  classical-classical or classical-quantum the duality for integrable systems. The most known examples are spectral duality between Gaudin model and Heisenberg Chain and spectral self-dual Toda lattice. The main difference is that spectral duality introduced by Morozov et. al. interchange $\lambda$ and $\mu$ on a spectral curve, for example Gaudin model and Heisenberg Chain spectral curves are connected in the following way
$$
\Gamma^{(\operatorname{Gaudin})}(\lambda, \mu) = \Gamma^{(\operatorname{Heisenberg})}(\mu, \lambda).
$$
So this duality may be viewed as Fourier transform between $\lambda$ and $\mu$. This Fourier transform duality  was firstly considered in series of works by Adams, Harnad, Hurtubise and Previato both for the isospectral and isomonodromic systems in order to provide the description of classical integrable systems in terms of coadjoint orbits of loop group action \cite{AHP88, AHH90, H94}. The main difference in out case that we do not have this twist of spectral parameters, since our duality comes not from the "Fourier" transform, but from the different reductions of the non-reduced systems.

\section{Dual Hamiltonians for multi-particle Painlev\'e-Calogero systems}

In this section we provide dual Hamiltonian systems which are obtained from Calogero-Painlev\'e IV,II and I by Hamiltonian reduction. In this section without loss of generality we will use $p_1,p_2\dots p_n$ and $q_1,q_2\dots q_n$ for the dual coordinates instead of  $I_1,I_2\dots I_n$ and $\phi_1,\phi_2\dots \phi_n$ since in this case they do not correspond to action-angle variables and we want to put more attention to the self-duality in case of $P_{\operatorname{IV}}$.
\subsection{Calogero-Painlev\'e IV}

\begin{proposition}\label{CPainIV}
There is an anti-symplectic involution of the reduced phase space for the Calogero-Painlev\`e IV  such that the reduced Hamiltonian is self-dual.
\end{proposition}
{\it Proof:}
The isomonodromic linear problem for the fourth matrix Painlev\'e system reads
\begin{equation}\label{2.1}
  \left\{\begin{array}{l}
    \frac{\partial}{\partial \lambda} \Psi = \left[\begin{array}{cc}
                                                      \frac{\mathbf{pq}}{\lambda} & \mathbf{q}\mathbf{p} +\theta_0 + \theta_1 -\frac{\mathbf{pqp}+\theta_0\mathbf{p}}{\lambda} \\
                                                       1 + \frac{\mathbf{q}}{\lambda}  & -\lambda + t + \frac{\mathbf{qp}+\theta_0}{\lambda}
                                                    \end{array}\right]\Psi\\ \\
  \frac{\partial}{\partial t}\Psi = \left[\begin{array}{cc}
                                                      i\frac{\lambda}{2} & \mathbf{q} \\
                                                      \mathbf{q} & -i\frac{\lambda}{2}
                                                    \end{array}\right]\Psi
                                                    \end{array}\right.
\end{equation}
The compatibility conditions are
\begin{equation}\label{2.2}
    \left\{\begin{array}{l}
           \dot{\mathbf{q}} = [\mathbf{p},\mathbf{q}]_{+} - \mathbf{q}^2 - t\mathbf{q}+\theta_0 \\
           \dot{\mathbf{p}} = [\mathbf{p},\mathbf{q}]_{+} - \mathbf{p}^2 + t \mathbf{p} + \theta_0 + \theta_1
         \end{array}\right.\quad
           H = \operatorname{Tr}\left( \mathbf{pq} (\mathbf{p}-\mathbf{q}-t) + \theta_0\mathbf{p}-(\theta_0+\theta_1)\mathbf{q} \right).
\end{equation}

In dual coordinates
\begin{equation}\label{2.3}
  \Lambda = \operatorname{diag}(p_1,p_2,...p_n),\quad \Phi=\operatorname{diag}(q_1,q_2,...q_n) - \left(\frac{ig}{p_i-p_j}\right)_{i\neq j}
\end{equation}
Hamiltonian reads
\begin{equation}\label{DIV}
  H_{\operatorname{IV}}^{(\operatorname{dual})} = \sum\limits_i\left[q_ip_i^2-p_iq_i^2-tq_ip_i+\theta_0p_i-(\theta_0+\theta_1)q_i\right] - g^2\sum\limits_{i<j}\frac{p_i+p_k}{(p_i-p_j)^2}
\end{equation}
Hence, we obtain that the change of variables 
\begin{displaymath}
  q_i\rightarrow -p_i,\quad p_i\rightarrow-q_i,\quad \theta_0\rightarrow \theta_1,\quad\theta_1\rightarrow\theta_0- \theta_1
\end{displaymath}
transforms $ H_{\operatorname{IV}}^{(\operatorname{dual})}$ to
\begin{equation}\label{CPIV}
  H_{\operatorname{IV}}^{(\operatorname{red})} = \sum\limits_i\left[q_ip_i^2-p_iq_i^2-tp_iq_i+\theta_0q_i-(\theta_0+\theta_1)p_i\right] + g^2\sum\limits_{i<j}\frac{q_i+q_k}{(q_i-q_j)^2}
\end{equation}
which is obtained by reduction at $(Q,P)$. This map changes the sign of the symplectic form, so we call it anti-symplectic involution. This coincides with the fact, that $P_{\operatorname{IV}}$ Hamiltonian from \eqref{2.2} invariant under the same change of variables in terms of $\mathbf{p}$ and $\mathbf{q}$, but symplectic form changes the sign.

\subsection{Calogero-Painlev\'e II}

We shall see that (in a contrast with the precedent model) the dual Hamiltonian of the Calogero-Painlev\`e II system doesn't "survive"under this duality map.

\begin{proposition}
The Calogero - Painlev\'e II system is {\it not} a self-dual. Moreover, the the dual system admits a {\it quadruple-wise} particle interaction, while in case of $\mathbf{q}$-diagonal reduction we obtain only the pair-wise particle interaction.
\end{proposition}
{\it Proof:} The isomonodromic problem for the matrix Painlev\'e II system is given by
\begin{equation}\label{1.1}
  \left\{\begin{array}{l} 
    \frac{\partial}{\partial \lambda} \Phi = \left[\begin{array}{cc}
                                                      i\frac{\lambda^2}{2} +i \mathbf{q}^2 + i\frac{t}{2}& \lambda\mathbf{q}-i\mathbf{p} -\frac{\theta}{\lambda} \\
                                                       \lambda\mathbf{q}+i\mathbf{p} -\frac{\theta}{\lambda}  & -i\frac{\lambda^2}{2} -i \mathbf{q}^2 - i\frac{t}{2}
                                                    \end{array}\right]\Phi\\ \\
  \frac{\partial}{\partial t} \Phi = \left[\begin{array}{cc}
                                                      i\frac{\lambda}{2} & \mathbf{q} \\
                                                      \mathbf{q} & -i\frac{\lambda}{2}
                                                    \end{array}\right]\Phi
                                                    \end{array}\right.
\end{equation}
which compatibility condition leads to the following equations
\begin{equation}\label{1.2}
  \left\{\begin{array}{l}
           \dot{\mathbf{q}} = \mathbf{p} \\
           \dot{\mathbf{p}} = 2\mathbf{q}^3 + t\mathbf{q} + \theta
         \end{array}\right.\quad
         H = \operatorname{Tr}\left(\frac{\mathbf{p}^2}{2} - \frac{1}{2}\left(\mathbf{q}^2+\frac{t}{2}\right)^2-\theta \mathbf{q}\right)
\end{equation}
In reduced coordinates $P$ and $Q$ the Hamiltonian reads
\begin{equation}\label{CPII}
H_{\operatorname{II}}^{(\operatorname{red})}=\sum_{j=1}^{n}\left[\frac{p_{j}^{2}}{2}-\frac{1}{2}\left(q_{j}^{2}+\frac{t}{2}\right)^{2}-\alpha q_{j}\right]+g^{2} \sum_{j \neq k} \frac{1}{\left(q_{j}-q_{k}\right)^{2}} 
\end{equation}
In dual coordinates Hamiltonian turns to
\begin{equation}\label{1.5}
\begin{gathered}
  H_{\operatorname{II}}^{(\operatorname{dual})} = \sum\limits_i\left[\frac{p_i^2}{2}-\frac{q_i^4}{2}-\frac{t}{2}q_i^2-\theta q_i\right] + 2g^2\sum\limits_{i<j}\frac{q_i^2+q_iq_j+q_j^2+\frac{t}{2}}{(p_i-p_j)^2}-2g^4\sum\limits_{i<j}\frac{1}{(p_i-p_j)^4}-\\
 -2g^4\left(\sum\limits_{i< j<k}\frac{1}{(p_i-p_j)^2(p_j-p_k)^2} + \sum\limits_{i< j<k<l}\frac{2}{(p_i-p_j)(p_j-p_k)(p_k-p_l)(p_l-p_i)} \right) 
  \end{gathered}
\end{equation}

 We shall postpone to the Appendix all technical details of explicit computations for  interaction terms in this Hamiltonian.
 
\subsection{Calogero-Painlev\'e I}
The isomonodromic problem for matrix $P_{\operatorname{I}}$ equation is given by
\begin{equation}\label{P1}
  \left\{\begin{array}{l} 
    \frac{\partial}{\partial \lambda} \Phi = \left[\begin{array}{cc}
                                                      \mathbf{p} & \lambda-\mathbf{q}\\
                                                      \lambda^2+\lambda\mathbf{q}+\mathbf{q}^2+\frac{t}{2}   & -\mathbf{p}
                                                    \end{array}\right]\Phi\\ \\
  \frac{\partial}{\partial t} \Phi = \left[\begin{array}{cc}
                                                      0 &  \frac{1}{2}\\
                                                      \frac{\lambda}{2}+\mathbf{q} & 0
                                                    \end{array}\right]\Phi \end{array}\right.
\end{equation}
gives the following Hamiltonian system
\begin{equation}\label{0.2}
 \left\{\begin{array}{ccl}
                                            \mathbf{\dot{q}} & = & \mathbf{p}   \\
                                            \mathbf{\dot{p}} & = & \frac{3}{2}\mathbf{q}^2+\frac{t}{2}
                                          \end{array}\right. \quad H = \operatorname{Tr}\left(\frac{\mathbf{p}^2}{2}-\frac{\mathbf{q}^3}{2} -\frac{t}{4}\mathbf{q}\right)
\end{equation}
In reduced coordinates $Q,P$ Hamiltonian is
\begin{equation}
H^{(\operatorname{red})}_{\operatorname{I}}=\sum_{j=1}^{n}\left[\frac{p_{j}^{2}}{2}- \frac{q_{j}^{3}}{2}-\frac{t}{4} q_{j}\right]+g^{2} \sum_{i<j} \frac{1}{\left(q_{i}-q_{j}\right)^{2}}
\end{equation}
Taking dual coordinates $\Lambda,\Phi$ we obtain
\begin{equation}
      H_{\operatorname{I}}^{(\operatorname{dual})}=\sum\limits_{i=0}^n\left[ \frac{p_i^2}{2}-\frac{q_i^3}{2}-\frac{t}{4}q_i\right] - \frac{3g^2}{2}\sum\limits_{j< i}\frac{q_i+q_j}{(p_i-p_j)^2}
\end{equation}
The obtained systems are not connected by anti-symplectic involution.

\begin{remark}
The self-dual Calogero-Painlev\'e IV system is completely integrable since its autonomous version is canonicaly equivalent to rational Inozemtsev system, which is completely integrable \cite{Takem1}. We don't know if the dual  Calogero-Painlev\'e I and II systems are completely integrable, this question requires further investigation and will be addressed elsewhere. 
\end{remark}

\section{Autonomous form of multi-particle Painlev\'e equations. Painlev\'e-Calogero correspondence}

In this section we write down Lax pairs for the autonomous versions of the Hamiltonians of noncommutative $P_{\operatorname{I}},\, P_{\operatorname{II}},\, P_{\operatorname{IV}}$. This procedure is nothing more than well-known Painlev\'e-Calogero correspondence for this type of equations. We use $\tau$ as parameter which doesn`t depend on time $t$ further in the text.

\subsection{Calogero-Painlev\'e II}

Autonomous form of matrix $P_{\operatorname{II}}$ has form
\begin{equation}\label{0.3}
\begin{gathered}
\left\{\begin{array}{ll}
    L(\mathbf{q},\mathbf{p})\Phi = \lambda \Phi \,\,&  L(\mathbf{q},\mathbf{p}; \lambda)= \mathcal{A}(\mathbf{q},\mathbf{p}, \lambda, \tau) = \left[\begin{array}{cc}
                                                      i\frac{\lambda^2}{2} +i \mathbf{q}^2 + i\frac{\tau}{2}& \lambda\mathbf{q}-i\mathbf{p} -\frac{\theta}{\lambda} \\
                                                       \lambda\mathbf{q}+i\mathbf{p} -\frac{\theta}{\lambda}  & -i\frac{\lambda^2}{2} -i \mathbf{q}^2 - i\frac{\tau}{2}
                                                    \end{array}\right] \\
    \frac{\partial}{\partial t} \Phi = M(\mathbf{q},\mathbf{p})\Phi \,\, & M(\mathbf{q},\mathbf{p};  \lambda)= \mathcal{B}(\mathbf{q},\mathbf{p},  \lambda, \tau)= \left[\begin{array}{cc}
                                                      i\frac{\lambda}{2} & \mathbf{q} \\
                                                      \mathbf{q} & -i\frac{\lambda}{2}
                                                    \end{array}\right]
\end{array}\right. \\
L_t + [L, M] = 0 \quad \Rightarrow \quad
                                                \left\{\begin{array}{l}
           \dot{\mathbf{q}} = \mathbf{p} \\
           \dot{\mathbf{p}} = 2\mathbf{q}^3 + \tau\mathbf{q} + \theta
         \end{array}\right.\\
         H = \operatorname{Tr}\left(\frac{\mathbf{p}^2}{2} - \frac{1}{2}\left(\mathbf{q}^2+\frac{\tau}{2}\right)^2-\theta \mathbf{q}\right)
\end{gathered}
\end{equation}

The moment map $[\mathbf{p}, \mathbf{q}]$ is still the first integral, so resticting to the level set
$$
  [\mathbf{p}, \mathbf{q}] = \sqrt{-1} g(\mathbf{1}-v^{T}\otimes v),\quad v=(1,1,...,1)
$$
and using reduced coordinates $P$ and $Q$ we get
$$
H_{II}^{(\operatorname{red})} = \sum_{j=1}^{n}\left[\frac{p_{j}^{2}}{2}-\frac{1}{2}\left(q_{j}^{2}+\frac{\tau}{2}\right)^{2}-\alpha q_{j}\right]+g^{2} \sum_{i<j} \frac{1}{\left(q_{i}-q_{j}\right)^{2}} 
$$
In dual coordinates $(\Phi, \Lambda)$ we obtain the following Hamiltonian system
\begin{displaymath}
\begin{gathered}
  H_{II}^{(\operatorname{dual})} = \sum\limits_i\left[\frac{p_i^2}{2}-\frac{q_i^4}{2}-\frac{\tau}{2}q_i^2-\theta q_i\right] +  2g^2\sum\limits_{i<j}\frac{q_i^2+q_iq_j+q_j^2+\frac{\tau}{2}}{(p_i-p_j)^2}-g^4\sum\limits_{i<j}\frac{1}{(p_i-p_j)^4}-\\ 
-2g^4\left(\sum\limits_{i<j<k}\frac{1}{(p_i-p_j)^2(p_j-p_k)^2} + \sum\limits_{i< j<k<l}\frac{2}{(p_i-p_j)(p_j-p_k)(p_k-p_l)(p_l-p_i)} \right)  
 \end{gathered}
\end{displaymath}
The Lax pair turns to
\begin{equation}\label{3.5}
\left\{\begin{array}{l}
  \tilde{L}\Psi = \lambda \Psi \\ \\
\frac{\partial}{\partial t} \Psi = (M- F\otimes\mathbf{1}_2)\Phi
\end{array}\right.
\end{equation}
where in coordinates $F$ in reduced and dual cases takes form
\begin{displaymath}
\begin{gathered}
  (q_i-q_j)^2F^{(\operatorname{red})}_{i,j} = ([A(Q,P),Q])_{i,j},\quad (p_i-p_j)^2F^{(\operatorname{dual})}_{i,j} = ([B(\Phi,\Lambda),\Lambda])_{i,j}\\
  F_{j,j}= -\sum\limits_{k\neq j}F_{j,k}+\frac{1}{n}\sum\limits_{l\neq k} F_{l,k}
\end{gathered}
\end{displaymath}
The $L$-matrix goes to
\begin{equation}
\begin{gathered}
L^{(\operatorname{red})}=P\otimes\sigma_2+\sqrt{-1}\left(\frac{\lambda^2}{2}+Q^2+\frac{t}{2}\right)\otimes\sigma_3 + \left(\lambda Q-\frac{\theta}{\lambda}\right)\otimes\sigma_1\\
L^{(\operatorname{dual})}=\Lambda\otimes\sigma_2+\sqrt{-1}\left(\frac{\lambda^2}{2}+\Phi^2+\frac{t}{2}\right)\otimes\sigma_3 + \left(\lambda\Phi-\frac{\theta}{\lambda}\right)\otimes\sigma_1
\end{gathered}
\end{equation}
where $\sigma_1, \sigma_2, \sigma_3$ are the Pauli matrices. 

\subsection{Calogero-Painlev\'e I}

"Freezing" $t$ in pair \eqref{P1} we obtain the following autonomous Hamiltonian system
\begin{equation}\label{3.1}
\begin{gathered}
\left\{\begin{array}{ll}
    L(\mathbf{q},\mathbf{p})\Phi = \lambda \Phi \,\,&  L(\mathbf{q},\mathbf{p}; \lambda)= \mathcal{A}(\mathbf{q},\mathbf{p}, \lambda, \tau) = \left[\begin{array}{cc}
                                                      \mathbf{p} &  \lambda-\mathbf{q}\\
                                                      z^2+z\mathbf{q}+\mathbf{q}^2+\frac{\tau}{2}   & -\mathbf{p}
                                                    \end{array}\right] \\
    \frac{\partial}{\partial t} \Phi = M(\mathbf{q},\mathbf{p})\Phi \,\, & M(\mathbf{q},\mathbf{p};  \lambda)= \mathcal{B}(\mathbf{q},\mathbf{p},  \lambda, \tau)=\left[\begin{array}{cc}
                                                      0 &  \frac{1}{2}\\
                                                      \frac{\lambda}{2}+\mathbf{q} & 0
                                                    \end{array}\right]
\end{array}\right. \\
L_t + [L, M] = 0 \quad \Rightarrow \quad \left\{\begin{array}{ccccl}
                                            \mathbf{\dot{q}} & = & A(\mathbf{p},\mathbf{q}) & = & \mathbf{p}   \\
                                            \mathbf{\dot{p}} & = & B(\mathbf{p},\mathbf{q}) & = &\frac{3}{2}\mathbf{q}^2+\frac{\tau}{2}
                                          \end{array}\right.\quad \\
                                                H^{\operatorname{(aut)}} = \operatorname{Tr}\left(\frac{\mathbf{p}^2}{2}-\frac{\mathbf{q}^3}{2}  -\frac{\tau}{4}\mathbf{q}\right)
\end{gathered}
\end{equation}
Restricting to the level set
$$
  [\mathbf{p}, \mathbf{q}] = \sqrt{-1} g(\mathbf{1}-v^{\mathbf{\operatorname{T}}}\otimes v),\quad v=(1,1,...,1)
$$
and then we use reduced coordinates $(Q,P)$
$$
H^{(\operatorname{red})}_{\operatorname{I}}=\sum_{j=1}^{n}\left[\frac{p_{j}^{2}}{2}-2 q_{j}^{3}-t q_{j}\right]+g^{2} \sum_{i < j} \frac{1}{\left(q_{i}-q_{j}\right)^{2}}
$$
In dual coordinates $(\Phi,\Lambda)$ we obtain the following Hamiltonian system
\begin{displaymath}
  H_I^{(\operatorname{dual})}=\sum\limits_{i=0}^n\left[ \frac{p_i^2}{2}-\frac{q_i^3}{2}-\frac{\tau}{4}q_i\right] - \frac{3g^2}{2}\sum\limits_{j< i}\frac{q_i+q_j}{(p_i-p_j)^2}
\end{displaymath}

The $L$-matrix has the following form
\begin{equation}\label{0.5}
\begin{gathered}
L^{(\operatorname{red})} = P\otimes\sigma_3+(\lambda-Q)\otimes\sigma_{+}+(\lambda^2+\frac{\tau}{2}+\lambda Q+Q^2)\otimes\sigma_{-}\\
  L^{(\operatorname{dual})} = \Lambda\otimes\sigma_3+(\lambda-\Phi)\otimes\sigma_{+}+(\lambda^2+\frac{\tau}{2}+\lambda\Phi+\Phi^2)\otimes\sigma_{-}
\end{gathered}
\end{equation}
where $\sigma_3,\sigma_{+},\sigma_{-}$ are the elements of the Cartan-Weyl basis of $\mathfrak{sl}(2,\mathbb{C})$.

\subsection{Painlev\'e-Calogero correspondence. Degeneration of rational Inozemtsev system}
Well known confluence scheme for the Painlev\'e equations \cite{OhOk06,Tak01} holds for the matrix Painlev\'e systems. From the point of view of reduced Calogero-Painlev\'e systems for the matrix Painlev\'e VI, V and IV these confluences are nothing but the nonautonomous version of degeneration for the Inozemtsev Hamiltonians. In the previous sections we show that corresponding autonomous multi-particle systems for Calogero-Painlev\'e I and II may be written in Lax form. According to Takasaki \cite{Takasaki1} multi-particle Calogero-Painlev\'e I and II systems should be further degenerations of the Inozemtsev rational Hamiltonians. 

This degeneration may be obtained by autonomization of confluence scheme for the Painlev\`e equations combined with the symplectic transformations given in \cite{BertolaCafassoRubtsov} to map reduced Hamiltonians to Inozemtsev "physical" Hamiltonians. In case of Calogero-Painlev\'e II and I we obtain physical Hamiltonians $\sum\limits_i \frac{p_i^2}{2}+V(q_i,t)$ by reduction without any additional canonical transforms. 

In this section we will use $\mathbf{q}^{(\operatorname{IV})}$ and  $\mathbf{q}^{(\operatorname{II})}$ for noncommutative canonical coordinates for the matrix Painlev\'e IV and the matrix Painlev\'e II systems respectively. We use the same notation for noncommutative moments and also for canonical cooordinates for reduced systems. 

\begin{theorem}
The conluence procedure from matrix Painlev\'e IV system (\ref{2.2}) to matrix Painlev\'e II system (\ref{1.2}) reduces to the cofluence Calogero-Painlev\'e multi-particle systems from (\ref{CPIV}) to (\ref{CPII}), but breaks for the dual systems (\ref{DIV}) and (\ref{1.5}).
\end{theorem}
\begin{proof}

To provide degeneration procedure from matrix Painlev\'e IV systems to matrix Painlev\'e II system we exploit confluence formula which is nothing but the symplectomorhism combined with rescaling and shift of Hamiltonian and time $t$ which contains parameter $\varepsilon$. Taking the limit $\varepsilon\rightarrow 0$ we obtain resulting Hamiltonian. 

To obtain (\ref{1.2}) from (\ref{2.2}) we use the following confluence symplectomorphism
\begin{equation}\label{conf}
\begin{gathered}
\mathbf{q}^{(\operatorname{IV})}=-\frac{1}{\varepsilon^3}\left(\frac{1}{2}+\varepsilon^2\mathbf{q}^{(\operatorname{II})}\right),\quad \mathbf{p}^{(\operatorname{IV})}=-\varepsilon \left(\mathbf{p}^{(\operatorname{II})}+\left(\mathbf{q}^{(\operatorname{II})}\right)^2+t/2\right) \\
t^{(\operatorname{IV})}=\frac{1}{\varepsilon^3}\left(1-\varepsilon^4t^{(\operatorname{II})}\right),\quad \theta_0 = -\frac{1}{4\varepsilon^6},\quad \theta_1=-\theta, \quad H^{(\operatorname{II})}=-\varepsilon H^{(\operatorname{IV})} + \varepsilon^{-2}\frac{\theta}{2}
\end{gathered}
\end{equation}
We see that (\ref{conf}) is indeed symplectomorphism

$$
\omega^{(\operatorname{IV})} = \operatorname{Tr}\operatorname{d}\mathbf{p}^{(\operatorname{IV})}\wedge \operatorname{d}\mathbf{q}^{(\operatorname{IV})} = \operatorname{Tr}\left(\operatorname{d}\mathbf{p}^{(\operatorname{II})}\wedge \operatorname{d}\mathbf{q}^{(\operatorname{II})}+\operatorname{d}\left(\mathbf{q}^{(\operatorname{II})}\right)^2\wedge \operatorname{d}\mathbf{q}^{(\operatorname{II})}\right) =\operatorname{Tr}\operatorname{d}\mathbf{p}^{(\operatorname{II})}\wedge \operatorname{d}\mathbf{q}^{(\operatorname{II})} = \omega^{(\operatorname{II})}
$$
since
\begin{multline*}
\operatorname{Tr}\left(\operatorname{d} \mathbf{q}^2\wedge \operatorname{d}\mathbf{q}\right)=\sum\limits_{i,k,l}\left(\operatorname{d}(q_{ik}q_{kl})\wedge \operatorname{d}q_{li}\right)
   = \sum\limits_{i,k,l}\underbrace{q_{ik}\operatorname{d}q_{kl}\wedge \operatorname{d}q_{li}}_\textrm{$i=\alpha,k=\beta,l=j$}+\underbrace{q_{kl}\operatorname{d}q_{ik}\wedge \operatorname{d}q_{li}}_\textrm{$k=\alpha, l=\beta , i=j$} 
  = \\ = \sum\limits_{\alpha,\beta}q_{\alpha\beta}\sum\limits_j\left[\operatorname{d}q_{\beta j}\wedge \operatorname{d}q_{j\alpha} + \operatorname{d}q_{j\alpha}\wedge \operatorname{d}q_{\beta j} \right] =0.
\end{multline*}
Besides rescaling of time agrees with the rescaling of Hamiltonian, such that Hamiltonian equations don't change. Under transformation (\ref{conf}) Hamiltonian  (\ref{2.2}) maps to the
\begin{multline}
H^{(\operatorname{IV})} = \operatorname{Tr}\left( \mathbf{p}^{(\operatorname{IV})}\mathbf{q}^{(\operatorname{IV})} (\mathbf{p}^{(\operatorname{IV})}-\mathbf{q}^{(\operatorname{IV})}-t^{(\operatorname{IV})}) + \theta_0\mathbf{p}^{(\operatorname{IV})}-(\theta_0+\theta_1)\mathbf{q}^{(\operatorname{IV})} \right) = \\ =
\varepsilon^2\operatorname{Tr}\left(\left(\mathbf{p}^{(\operatorname{II})}\right)^2\mathbf{q}^{(\operatorname{II})}+(2\mathbf{p}^{(\operatorname{II})}+t^{(\operatorname{II})})\left(\mathbf{q}^{(\operatorname{II})}\right)^3+\left(\mathbf{q}^{(\operatorname{II})}\right)^5+\mathbf{p}^{(\operatorname{II})}\mathbf{q}^{(\operatorname{II})}t + \frac{t^{(\operatorname{II})}}{4}\mathbf{q}^{(\operatorname{II})}\right) + \\
+\operatorname{Tr}\left(\frac{\left(\mathbf{p}^{(\operatorname{II})}\right)^2}{2} - \frac{1}{2}\left(\left(\mathbf{q}^{(\operatorname{II})}\right)^2+\frac{t^{(\operatorname{II})}}{2}\right)^2-\theta \mathbf{q}^{(\operatorname{II})}\right)
\end{multline}
After taking the limit  $\varepsilon\rightarrow 0$ and putting down the superscripts we get
the Hamiltonian matrix Painlev\'e II system
$$
H = \operatorname{Tr}\left(\frac{\mathbf{p}^2}{2} - \frac{1}{2}\left(\mathbf{q}^2+\frac{t}{2}\right)^2-\theta \mathbf{q}\right).
$$
This transformation reduces to the multi-particle Calogero-Painlev\'e IV system
\begin{multline*}
H_{\operatorname{IV}}^{(\operatorname{red})} = \sum\limits_{i=1}^n\left[q^{(\operatorname{IV})}_i(p^{(\operatorname{IV})}_i)^2-p^{(\operatorname{IV})}_i(q^{(\operatorname{IV})}_i)^2-tp^{(\operatorname{IV})}_iq^{(\operatorname{IV})}_i+\theta_0q^{(\operatorname{IV})}_i-(\theta_0+\theta_1)p^{(\operatorname{IV})}_i\right] +\\+ g^2\sum\limits_{i<j}\frac{q^{(\operatorname{IV})}_i+q^{(\operatorname{IV})}_k}{\left(q^{(\operatorname{IV})}_i-q^{(\operatorname{IV})}_j\right)^2}
\end{multline*}
by the following way 
$$
q^{(\operatorname{IV})}_i=-\frac{1}{\varepsilon^3}\left(\frac{1}{2}+\varepsilon^2q^{(\operatorname{II})}_i\right),\quad p^{(\operatorname{IV})}_i=-\varepsilon \left(p^{(\operatorname{II})}_i+(q_i^{(\operatorname{II})})^2+\frac{t}{2}\right)
$$
and the same rescaling for the constants, time and Hamiltonian as in (\ref{conf}). Since the "diagonal" part of Calogero-Painlev\'e IV Hamiltonian is a sum of uncoupled Painlev\'e IV Hamiltonians, in the limit $\varepsilon\rightarrow 0$  transform will map them to the sum of uncoupled Painlev\'e Hamiltonians. For the interaction term we have
$$
-\varepsilon g^2\sum\limits_{i<j}\frac{q^{(\operatorname{IV})}_i+q^{(\operatorname{IV})}_j}{(q^{(\operatorname{IV})}_i-q^{(\operatorname{IV})}_j)^2}=\sum\limits_{i<j}\varepsilon^2g^2\frac{q^{(\operatorname{II})}_i+q^{(\operatorname{II})}_j}{\left(q^{(\operatorname{II})}_i-q^{(\operatorname{II})}_j\right)^2}+\frac{g^2}{\left(q^{(\operatorname{II})}_i-q^{(\operatorname{II})}_j\right)^2}\rightarrow \sum\limits_{i<j}\frac{g^2}{\left(q^{(\operatorname{II})}_i-q^{(\operatorname{II})}_j\right)^2}
$$
which transforms to the to the interaction term for the Calogero-Painlev\'e II system. However this reduction of the symplectomorphism (\ref{conf}) cannot be applied for the dual systems. Indeed this confluence on the Calogero-Painlev\'e side is a consequence of linearity of transformation (\ref{conf}) with respect to $\mathbf{q}$, we see that eigenspace of $\mathbf{q}^{(\operatorname{IV})}$ coincides with the eigenspace of $\mathbf{q}^{(\operatorname{II})}$, so in some sense we reduce from the same point of $GL$-orbit. On the other hand, on the dual side we have that dual reduction (diagonalization of $\mathbf{p}^{(\operatorname{IV})}$) for the matrix Painlev\'e IV system implies diagonalization of $\mathbf{p}^{(\operatorname{II})}+\left(\mathbf{q}^{(\operatorname{II})}\right)^2$ which doesn't coincide with the reduction at diagonal $\mathbf{p}^{(\operatorname{II})}$, so on the dual side we have that dual systems reduce from the different points of $GL$-orbit, which is an obstacle for this degeneration for the dual systems.
\end{proof}

\begin{remark}
Here we use the combination of the polynomial canonical transformation with confluence transformation given in \cite{OhOk06,Tak01} which is "linear" both for $\mathbf{q}$ and $\mathbf{p}$ and given by
\begin{equation}\label{conf1}
\begin{gathered}
\mathbf{q}^{(\operatorname{IV})}=-\frac{1}{\varepsilon^3}\left(\frac{1}{2}+\varepsilon^2\tilde{\mathbf{q}}^{(\operatorname{II})}\right),\quad \mathbf{p}^{(\operatorname{IV})}=-\varepsilon \tilde{\mathbf{p}}^{(\operatorname{II})}, \quad t^{(\operatorname{IV})}=\frac{1}{\varepsilon^3}\left(1-\varepsilon^4t^{(\operatorname{II})}\right)\\
\theta_0 = -\frac{1}{4\varepsilon^6},\quad \theta_1=-\theta, \quad \tilde{H}^{(\operatorname{II})}=-\varepsilon H^{(\operatorname{IV})} + \varepsilon^{-2}\frac{\theta}{2}
\end{gathered}
\end{equation}
and maps matrix Painlev\'e IV system to the following Hamiltonian
\begin{equation}\label{polP}
 \tilde{H}^{(\operatorname{II})} = \operatorname{Tr}\left(\frac{1}{2}\tilde{\mathbf{p}}(\tilde{\mathbf{p}}-2\tilde{\mathbf{q}}^2-t)-\theta\tilde{\mathbf{q}}\right)
\end{equation}
which Calogero-Painlev\'e reduction takes form
$$
\sum\limits_i \left[\frac{p_i^2}{2} - p_iq_i^2-\frac{t}{2}p_i-\theta p_i \right] + g^2\sum\limits_{i< j} \frac{1}{(q_i-q_j)^2}
$$
and the dual reduction is
$$
\sum\limits_i \left[\frac{p_i^2}{2} - p_iq_i^2-\frac{t}{2}p_i-\theta q_i \right] - g^2\sum\limits_{i\neq j} \frac{p_i}{(p_i-p_j)^2}.
$$
Since (\ref{conf1}) is linear in $\mathbf{q}$ and $\mathbf{p}$ this confluence holds also for Painlev\'e-Calogero and the dual systems. Hamiltonian system $\tilde{H}^{(\operatorname{II})}$ is canonicaly equivalent to the (\ref{1.2}) with the symplectic transformation
\begin{equation}\label{canTr}
\tilde{\mathbf{q}}=\mathbf{q},\quad \tilde{\mathbf{p}}=\mathbf{p}+\mathbf{q}^2+t/2
\end{equation}
which is obstacle for the degeneration of dual systems. The degeneration to matrix Painlev\'e I systems is given by "linear" maps in $\mathbf{q}$ and $\mathbf{p}$ from the $\tilde{H}^{(\operatorname{II})}$ system. To obtain confluence from $H^{(\operatorname{II})}$ given in (\ref{1.2}) we need to use inversion of (\ref{canTr}) combined with confluence which leaves transformation to be "linear" only in $\mathbf{q}$. So we have that degeneration to matrix Painlev\'e I system reduces to the Calogero-Painlev\'e systems but not for the dual systems.
\end{remark}

By fixing $t=\tau$ we get the autonomous version of the confluence which provides the degeneration of the rational Inozemtsev system to the rational system obtained by the reduction of autonomous matrix Painlev\'e II and I systems.

Obtained multi-particle systems are further degenerations of the Inozemtsev systems. Arising dual systems look new and need more detailed investigation. In case of the multi-particle Calogero-Painlev\'e II there is an interesting interpretation of this deautonomisation which arises from the reduction from matrix modified Korteweg-de Vries equation. In the next section we discuss this connection in details.

\section{Matrix modified Korteweg-de Vries equation and Painlev\'e-Calogero correspondence for Painlev\'e II}

The matrix modified Korteweg-De Vries equation (MmKdV) has the form \cite{KhaKhr}
\begin{equation}\label{4.1}
\mathbf{u}_t=\mathbf{u}_{xxx}+3[\mathbf{u}_x,\mathbf{u}]-6\mathbf{u}\mathbf{u}_x\mathbf{u}.
\end{equation}
In this section we are going to show the connection between travelling wave reduction of MmKdV and it's self-similar reduction. 
Self-similar reduction is the following change of variables
\begin{equation}
    \mathbf{u}(x,t)=\frac{\mathbf{v}(z)}{(3t)^{\frac{1}{3}}},\quad z = \frac{x}{(3t)^{\frac{1}{3}}}.
\end{equation}
This change of variables leads to
\begin{displaymath}
\partial_t = -\frac{x}{(3t)^{\frac{4}{3}}}\partial_z,\quad \partial_x = \frac{1}{(3t)^{\frac{1}{3}}}\partial_z
\end{displaymath}
and the equation \eqref{4.1} turns to
\begin{equation}\label{4.2}
    \mathbf{v}_{zzz} +3[\mathbf{v},\mathbf{v}_{zz}]-6\mathbf{v}\mathbf{v}_z\mathbf{v}+2(z\mathbf{v})_z=0.
\end{equation}
Let us change since that moment variable $z$ on $t$ in order to return to our previous notation. It can be shown that equation \eqref{4.2} is just a differential consequence of the matrix Painlev\'e II \cite{BalandinSokolov, GordPickZhu}
\begin{equation}\label{4.3}
    \left(\partial_t + 3 ad_{\mathbf{v}} + 2 ad_{\mathbf{v}}\circ \partial_t^{-1}\circ ad_{\mathbf{v}}\right)\circ \left(\mathbf{v}_{tt} + 2\mathbf{v}^3+t\mathbf{v}+\theta\right) = 0
\end{equation}
where $ad_{\mathbf{v}}\circ A = [\mathbf{v}, A]$. Equation \eqref{4.3} is matrix Painlev\'e II. On the other hand, the travelling wave reduction of the MmKdV equation
\begin{equation}
    \mathbf{u}(x,t) = \mathbf{v}(z),\quad z = x+\omega t
\end{equation}
has the form
\begin{equation}\label{4.3}
    \left(\partial_z + 3 ad_{\mathbf{v}} + 2 ad_{\mathbf{v}}\circ \partial_z^{-1}\circ ad_{\mathbf{v}}\right)\circ \left(\mathbf{v}_{zz} + 2\mathbf{v}^3+\omega\mathbf{v}+\theta\right) = 0
\end{equation}
where $\omega$ is an arbitrary constant. Let us change $z$ to $t$ and $\omega$ to $\tau$ in order to show the connection with equation \eqref{4.3}. The equation turns to
\begin{equation}
    \mathbf{v}_{tt} + 2\mathbf{v}^3+\tau\mathbf{v}+\theta=0
\end{equation}
and can be written as the following Hamiltonian system
\begin{equation}
      \left\{\begin{array}{l}
           \dot{\mathbf{q}} = \mathbf{p} \\
           \dot{\mathbf{p}} = 2\mathbf{q}^3 + \tau\mathbf{q} + \theta
         \end{array}\right.,\quad H_{\operatorname{PII}} (\mathbf{q},\mathbf{p}, \tau) = \operatorname{Tr}\left(\frac{\mathbf{p}^2}{2} - \frac{1}{2}\left(\mathbf{q}^2+\frac{\tau}{2}\right)^2-\theta \mathbf{q}\right)
\end{equation}
As a consequence, we have that matrix Painlev\'e II (which is self-similar reduction of MmKdV) is a $\tau$-deformation of the travelling wave reduction of MmKdV. This Painlev\'e-Calogero correspondence can be lifted down to the multi-particle systems which are obtained by Hamiltonian reduction, and finally we have the following diagram

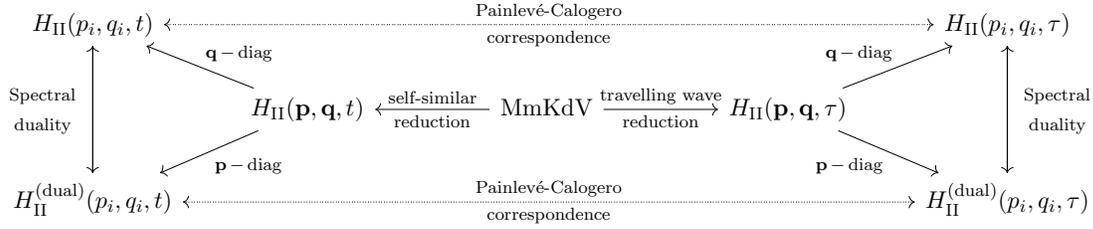
\begin{center}
\begin{tikzpicture}[baseline= (a).base]
\node[scale=0.83] (a) at (0,0){
\begin{tikzcd}
 H_{\operatorname{II}}(p_i, q_i, t) \arrow[dd, "\begin{array}{c}
      \text{Spectral}\\
      \text{duality}
 \end{array}"'] \arrow[rrrrrr, dotted, "\text{Painlev\'e-Calogero}"] & & & & &  & H_{\operatorname{II}}(p_i, q_i, \tau)   \arrow[dd,"\begin{array}{c}
      \text{Spectral}\\
      \text{duality}
 \end{array}"] \arrow[llllll, dotted, "\text{correspondence}"]  \\
 & H_{\operatorname{II}}(\mathbf{p},\mathbf{q}, t) \arrow[lu, "\mathbf{q} \,-\, \operatorname{diag}"'] \arrow[ld,"\mathbf{p}\,-\,\operatorname{diag}"] & & \operatorname{MmKdV}  \arrow[ll,"\operatorname{self-similar}"', "\operatorname{reduction}"] \arrow[rr, "\text{travelling wave}","\operatorname{reduction}"'] & &  H_{\operatorname{II}}(\bf{p},\bf{q}, \tau)  \arrow[rd, "\mathbf{p} \,-\, \operatorname{diag}"'] \arrow[ur, "\mathbf{q} \,-\, \operatorname{diag}"] & \\
 H^{(\operatorname{dual})}_{\operatorname{II}}(p_i, q_i, t)   \arrow[rrrrrr, dotted, "\text{Painlev\'e-Calogero}"] \arrow[uu] & & & & & &  H^{(\operatorname{dual})}_{\operatorname{II}}(p_i, q_i, \tau)  \arrow[llllll, dotted, "\text{correspondence}"] \arrow[uu]\\
\end{tikzcd}};
\end{tikzpicture}
\end{center}

The following situation is the result that powers of $t$ in self-similar variable and the multiplier of $\mathbf{u}$ coincides, so in a sense this phenomena is a result of existence of special symmetries of integrable PDE. 
\begin{theorem}
For all integrable evolution PDE's which allows the self-similar variables of the form
$$
u(x,t) = \frac{v(z)}{t^\beta}, \quad z=\frac{x}{t^\beta}
$$
and also allows travelling wave reduction
$$
u=w(z),\quad z=x-\omega  t
$$
self-similar reduction is $\omega$-deformation of travelling wave reduction
\end{theorem}
\begin{proof}
By integrable evolution PDE we mean the following equation
$$
u_t + \partial_x F(x, u, u_x, u_{xx},...) = 0.
$$
The existence of the travelling wave solution, means that equation is invariant under translations, so we have that $F$ doesn't depend on $x$. Since $F$ depends only on $u$ and it's derivatives with respect to $x$, the existence of the self similar variables
$$
u = \frac{v(z)}{t^{\beta}} ,\quad z = \frac{x}{t^{\alpha}}
$$
means that $F$ transforms as homogeneous function and goes to the $\frac{F(v,v_z,...)}{t^{\alpha+1}}$. The reduction takes form
$$
\beta(v+\frac{\alpha}{\beta}zv_z)-\partial_z F(v,v_z,...) = 0
$$
which in case of $\alpha=\beta$ may be integrated to
$$
F(v,v_z,v_{zz}...) - \beta z v = C,
$$
where $C$ is a constant of integration. On the other hand travelling wave reduction $u=w(z=x-\omega  t)$ takes form
$$
F(v,v_z,v_{zz}...) - \beta \omega  v = 0
$$
and we see that self-similar reduction is deformation of travelling wave reduction with respect to $\omega$.
\end{proof}

\section{Further remarks and open questions}
This paper is a small first step in the study of dualities for non-autonomous many-particle systems. Our natural examples and computations put a natural question  of {\it Liouville  and quantum integrability} for obtained dual rational many-particle systems. Some of them are integrable {\it ad hoc}, but integrability  of other examples is not clear at all.

We should remark that in all our examples of the Hamiltonian reduction procedure for Calogero-Painlev\'e Hamiltonian systems we always use for their phase space the
cotangent bundle of the Lie algebrs ${\mathfrak gl}_n$ and the momentum mapping given by the matrix commutator $\mu = [P,Q].$ In the same time the Hamiltonian reductions for their Calogero-Moser prototypes can be described with various generalizations of this phase space and the  momentum map: group-like half-commutator  $\mu = PQP^{-1} - Q$ on $T^*GL_n$ or "full"group-commutator $\mu =  PQP^{-1}Q^{-1}$ on the Heisenberg double $G\times G$. The existence of isomonodromic systems on  on $T^*GL_n$ and Heisenberg double $G\times G$ is open problem, which leads to the following natural question: how can we extend our dualities to the trigonometric and elliptic Calogero-Painlev\'e Hamiltonians? The Calogero-Moser-Ruijsenaars-Schneider duality table dictates a necessity of existence for some "Ruijsenaars-Painlev\'e correspondence" where on the Painlev\'e side one should expect an appearance of some discrete or $q-$Painlev\'e systems.

A strong evidence to the existence of this speculative extension is based on the above discussed one-to-one correspondence between so-called $BC_n$ Inozemtsev 
model with the Hamiltonian $H_{BC_n}$

\begin{equation*}
 H_{BC_{n}} : \quad\sum_{j=1}^{n}\left(\frac{p_{j}^{2}}{2} - \sum_{\ell=0}^{3} \kappa_{\ell}(\kappa_{\ell} +1) \wp\left(q_{j}+\omega_{\ell}\right)\right) - \frac{\kappa(\kappa +1)}{2} \sum_{j \neq k}\left(\wp\left(q_{j}-q_{k}\right)+\wp\left(q_{j}+q_{k}\right)\right)
\end{equation*}

It is known (see e.g. \cite{Takem1}) that  the $BC_n$ Inozemtsev model is a universal completely integrable model of quantum mechanics and a correspondence between the $BC_n$ Ruijsenaars-van Diejen systems  and the $BC_n$ Inozemtsev  was established. One should take into account the recent result of the same author \cite{Takem2}, who has studied an analogue of the well-known  relationship between Painlev\'e VI and Heun differential systems for the case of  difference equations. He has proposed, in particular, a correspondence between elliptic difference Painlev\'e equations and the one-variable Ruijsenaars-van Diejen difference equation, regarded  as a difference analogue of the Heun equation. He has proved also that  (degenerated) Rujisenaars-van Diejen operators of one variable are special cases of linear $q-$difference equations related to certain q-Painlev\'e VI equations by a connection preserving deformation.

It would be extremely interesting to find an appropriate version of Hamiltonian reduction and to find an analogue of the Ruijsenaars duality in the framework of this conjectural "Ruijsenaars-Painlev\'e correspondence",  which should include, a correspondence between difference systems like elliptic Ruijsenaars, dual to elliptic Calogero-Moser and fabulous double-elliptic (DELL) and elliptic Rains Painlev\'e system and its various degenerations. The first interesting results in this direction were very recently obtained by
Noumi, Ruijsenaars and Yamada \cite{NRY2019} where they have shown that the 8-parameter elliptic Sakai difference Painlev\'e equation can be presented in a Lax-like form which can be specified as the non-autonomous gives Schr\"odinger equation equation for the $BC_1$ 8-parameter Ruijsenaars-van Diejen difference Hamiltonian.
 
Another intriguing and challenging problem concerns a transition of the described dualities to other close domains. We suppose that there are some interesting links between our classical dualities and their quantum counterparts, which we did not touch in this paper. In particular, we expect the existence of a close connection between quantum version of our dualities and dualities described by Koornwinder and Mazzocco \cite{KM} in the $q-$Askey scheme and the degenerate DAHA. We are going to clarify it and to fill some gaps in one of the tables from \cite{CMR2019}. In this direction it would be highly ambitious and interesting to understand the conjectural correspondence in the case trigonometric degeneration of the difference Ruijsenaars-Shneider systems and its duality with $q-$ Knizhnik-Zamolodchikov equations or, more generally, in the framework of $q-$ Langlands correspondence (see, e.g. \cite{KSZ2018}).

We close a survey of possible applications of our computational observations with more general open question related to the end of Section 5. This is the condition for symmetries of PDE which lead to the situation when one reduction is deformation of another. Besides for integrable PDE's lifting of this symmetries to the Lax pairs is also an essential issue in this problem. We also believe that these conditions for the self-similar reduction may be lifted to the non-commutative case and more general classes of equations. A part of these questions (for the Sawada-Kotera hierarchy) will be addressed in PhD Thesis of Irina Bobrova  (National Research University HSE) (see \cite{IB19}).

\begin{appendices}
\section{Calculation of interaction terms}
In case of multi-particle dual Painlev\'e systems we obtain Hamiltonians which contain terms with more than two particle interactions. It comes from the $Q^3$ and $Q^4$ terms in matrix Hamiltonians, where $Q$ is the rational Calogero Lax operator
$$
Q = \delta_{ij} q_{i} + (1-\delta_{ij})\frac{\sqrt{-1}g}{p_i-p_j}.
$$

In case of fixed size of matrix $Q$ (the number of particles) calculation of traces of any power of $Q$ is a straightforward problem, which may be solved for example with help of computer algebra systems. In case of arbitrary size of $Q$, this computation is also not a big problem, but the process may be tedious. 

Here we present a simple approach for a calculation of such interaction terms. Here we denote by $n$ a number of particles (the size of $Q$) and by $g$ a coupling constant. Since $Q$ is linear matrix function of $g$, the trace of $Q^l$ is a polynomial in $g$ of degree less than $l$ 
\begin{equation}\label{a1}
\operatorname{Tr}Q^{l} = \sum\limits_{k=0}^l \frac{g^k}{k!}F_k,\quad F_k = \frac{d^k}{d g^k}\operatorname{Tr}(Q^{l}) =  \operatorname{Tr}\left(\frac{d^k}{d g^k}Q^{l}\Big\vert_{g=0}\right).
\end{equation}
To compute coefficients we use the following technical lemma
\begin{lemma}
The polynomial $\operatorname{Tr}Q^{l}$ is even for all $l$, i.e.
\begin{equation}
\operatorname{Tr}Q^{l} = \sum\limits_{k=0}^{\left[\frac{l}{2}\right]} \frac{g^{2k}}{2k!}F_{2k}
\end{equation}
where $[\frac{l}{2}]$ equals to $\frac{l}{2}-\frac{1}{2}$ in case when $l$ is odd and $\frac{l}{2}$ when $l$ is even.
\end{lemma}
\begin{proof}
In further calculations we will use the following notation
\begin{equation}
    Q\Big\vert_{g=0} = \delta_{ij} q_{i} = D,\quad \frac{dQ}{dg}= (1-\delta_{ij})\frac{\sqrt{-1}}{p_i-p_j} = A.
\end{equation}
Each odd coefficient $F_{2r-1}$ of polynomial (\ref{a1}) takes form
\begin{equation}\label{a2}
F_{2r-1} = \sum\limits_{0<i_1<\dots<i_{2r-1}\leq l}\operatorname{Tr}E_{i_1,i_2\dots i_{2r-1}}
\end{equation}
where $E$ is given by
$$
E_{i_1,i_2\dots i_{2r-1}}=D^{i_1-1}AD^{i_2-i_1-1}A\dots D^{i_k-i_{k-1}-1}A^{i_k}\dots  D^{i_{2r-1}-i_{2r-2}-1}AD^{l-i_{2r-1}}
$$
Since $A$ is symmetric matrix, $D$ is skew-symmetric and all $E$'s contain the odd number of $D$'s we have
$$
(E_{i_1,i_2,\dots i_{2r-1}})^{\operatorname{T}} = - E_{l-i_{2r-1},i_{2r-2},\dots i_{2r-i_1}}
$$
On the other hand, trace is invariant under transposition, so there are two cases. The first one is
$$
E_{i_1,i_2,\dots i_{2r-1}}\neq E_{l-i_{2r-1},i_{2r-2},\dots i_{2r-i_1}}
$$
so for each sequence $(i_1,i_2,\dots i_{2r-1})$ in sum (\ref{a2}) we have the "mirror" one $(l-i_{2r-1},l-i_{2r-2},$ $\dots, l-i_{2r-i_1})$, such that
$$
\operatorname{Tr}E_{i_1,i_2,\dots i_{2r-1}} + \operatorname{Tr} E_{l-i_{2r-1},i_{2r-2},\dots i_{2r-i_1}} = 0
$$
so they don't contribute to (\ref{a1}). The second case is
$$
E_{i_1,i_2,\dots i_{2r-1}} = E_{l-i_{2r-1},i_{2r-2},\dots i_{2r-i_1}}.
$$
In that case $E_{i_1,i_2,\dots i_{2r-1}}$ is a skew-symmetric matrix, so it doesn't contribute in (\ref{a2}).
\end{proof}
Finally we provide calculations of traces for $Q^3$ and $Q^4$. In case of $l=3$ we have
$$
\operatorname{Tr}Q^{3} = \operatorname{Tr}(D^3) + 3 g^2 \operatorname{Tr}(AAD) = \sum\limits_{i=1}^n q_i^3 + 3g^2\sum\limits_{i\neq j}\frac{q_i}{(p_i-p_j)^2} = \sum\limits_{i=1}^n q_i^3 + 3g^2\sum\limits_{i < j}\frac{q_i+q_j}{(p_i-p_j)^2}
$$
For the $l=4$ the expansion takes form
$$
\operatorname{Tr}Q^{4} =  \operatorname{Tr}(D^4) + 2g^2\left[2\operatorname{Tr}(D^2A^2)+\operatorname{Tr}(DADA)\right] + g^4\operatorname{Tr}(A^4)
$$
Here we have
$$
\operatorname{Tr}(D^2A^2) = \sum\limits_{i\neq j} \frac{q_i^2}{(p_i-p_j)^2} =  \sum\limits_{i<j} \frac{q_i^2+q_j^2}{(p_i-p_j)^2}, \quad \operatorname{Tr}(DADA) = \sum\limits_{i\neq j} \frac{q_iq_j}{(p_i-p_j)^2} =  2\sum\limits_{i<j} \frac{q_iq_j}{(p_i-p_j)^2} 
$$
The last $g^4$-term takes form
$$
\operatorname{Tr}(A^4) = \sum\limits_{i\neq j \neq k\neq l\neq i}\frac{1}{(p_i-p_j)(p_j-p_k)(p_k-p_l)(p_l-p_i)}
$$
and we have 3 opportunities ($k=i,l=j$), ($l=j$) and all indices $i,j,k,l$ are different, which leads to the following interaction term 
$$
\operatorname{Tr}(A^4)= \sum\limits_{i<j}\frac{2}{(p_i-p_j)^4}+\sum\limits_{i< j<k}\frac{4}{(p_i-p_j)^2(p_j-p_k)^2} + \sum\limits_{i< j<k<l}\frac{8}{(p_i-p_j)(p_j-p_k)(p_k-p_l)(p_l-p_i)}\,.
$$

\end{appendices}

\end{document}